\documentclass[12pt, journal, onecolumn]{IEEEtran}
\IEEEoverridecommandlockouts

\usepackage{amsmath}
\usepackage{longtable}
\usepackage{rotating}
\usepackage{multirow}
\usepackage{epsfig, amsmath,amssymb,epsf,cite,subfigure,scalefnt,multirow,array,setspace}
\usepackage{algorithm}
\usepackage{algorithmic}
\usepackage{epstopdf}
\usepackage{subfig}
\usepackage{amsfonts}
\usepackage{booktabs}
\usepackage{multirow}
\usepackage{multicol}
\usepackage{cite}
\usepackage{algorithm}
\usepackage{algorithmic}
\usepackage{stfloats}
\usepackage{bm}
\usepackage{graphicx}
\usepackage{color}
\usepackage{enumerate}
\usepackage{url}

\usepackage[final]{pdfpages}
\usepackage{fancyhdr}
\usepackage{lineno}
\allowdisplaybreaks[4]

\newcommand{\redd}[1]{\textcolor{black}{#1}}

\def\cA{{\mathcal{A}}}  \def\cC{{\mathcal{C}}} \def\cD{{\mathcal{D}}}

 \def\cN{{\mathcal{N}}}  
  \def\cS{{\mathcal{S}}} 
\def\cU{{\mathcal{U}}}

\def\ba{{\mathbf{a}}} \def\bb{{\mathbf{b}}}   
\def\bff{{\mathbf{f}}}    
   \def\bn{{\mathbf{n}}} 
   \def\bs{{\mathbf{s}}} 
\def\bu{{\mathbf{u}}} \def\bv{{\mathbf{v}}}  \def\bx{{\mathbf{x}}} \def\by{{\mathbf{y}}}

\def\bA{{\mathbf{A}}} \def\bB{{\mathbf{B}}}   
\def\bF{{\mathbf{F}}} \def\bG{{\mathbf{G}}} \def\bH{{\mathbf{H}}} \def\bI{{\mathbf{I}}} 
 \def\bL{{\mathbf{L}}}   
\def\bP{{\mathbf{P}}}  \def\bR{{\mathbf{R}}}  
\def\bU{{\mathbf{U}}} \def\bV{{\mathbf{V}}} \def\bW{{\mathbf{W}}} \def\bX{{\mathbf{X}}}

\def\argmin{\mathop{\mathrm{argmin}}}
\def\argmax{\mathop{\mathrm{argmax}}}

\def\tr{\mathop{\mathrm{tr}}}

\def\rank{\mathop{\mathrm{rank}}}
\def\span{\mathop{\mathrm{span}}}
\def\vec{\mathop{\mathrm{vec}}}

\def\diag{\mathop{\mathrm{diag}}}

\def\mod{\mathop{\mathrm{mod}}}

\renewcommand{\baselinestretch}{1.0}

     \def\d4{\!\!\!\!}              
                  \def\gam{\gamma}     \def\al{\alpha}
    \def\del{\delta}           
\def\bSig{\mathbf{\Sigma}}   \def\bLam{\mathbf{\Lambda}}
 \def\bPhi{\mathbf{\Phi}}   \def\lam{\lambda} 
 
  \def\bPsi{\mathbf{\Psi}}
\def\bphi{\boldsymbol{\mathop{\mathrm{\phi}}}} \def\bpsi{\boldsymbol{\mathop{\mathrm{\psi}}}}

  \def\R{{\mathbb{R}}} \def\C{{\mathbb{C}}}   \def\E{{\mathbb{E}}}

\def\lp{\left(}     \def\rp{\right)}            \def\lS{ \left[ }
\def\rS{ \right] }  \def\la{\left|}     \def\ra{\right|}    \def\lA{\left\|}     \def\rA{\right\|}

  \def\sig{\sigma}     \def\tbn{\tilde{\bn}}

  \def\-{\! - \!}  \def\+{\! + \!}  \def\={\! = \!}  \def\>{\! > \!} \def\nn{\nonumber}

\def \log{\mathrm{log}}
\def\exp{\mathrm{exp}}

\newtheorem{theorem}{Theorem}
\newtheorem{lemma}{Lemma}

\newtheorem{remark}{Remark}

\newcommand{\bef}{\begin{figure}}
\newcommand{\eef}{\end{figure}}
\newcommand{\beq}{\begin{eqnarray}}
\newcommand{\eeq}{\end{eqnarray}}

\newenvironment{proof}[1][Proof]{\begin{trivlist}
\item[\hskip \labelsep {\bfseries #1}]}{\end{trivlist}}

\newcommand{\qed}{\nobreak \ifvmode \relax \else
\ifdim\lastskip<1.5em \hskip-\lastskip \hskip1.5em plus0em
minus0.5em \fi \nobreak \vrule height0.5em width0.5em
depth0.25em\fi}

\begin{document}

\pagenumbering{arabic}

\title{Leveraging the Restricted Isometry Property: Improved Low-Rank Subspace Decomposition for Hybrid Millimeter-Wave Systems }
\author{Wei Zhang, Taejoon Kim, David J. Love, and Erik Perrins
\thanks {Parts of this work were previously presented
at the IEEE Global Communications Conference, Washington, DC USA 2016 \cite{ZhangWei}.
D. J. Love was supported in part by the National Science Foundation (NSF) under CNS1642982.

{W. Zhang is with the Department of Electronic Engineering, City University of Hong Kong , Kowloon, Hong Kong, China (e-mail: wzhang237-c@my.cityu.edu.hk).}
{T. Kim and E. Perrins are with the Department of Electrical Engineering and Computer Science, University of Kansas, KS 66045, USA (e-mail: taejoonkim@ku.edu, perrins@ku.edu).}
{D. J. Love is with the School of Electrical and Computer Engineering, Purdue University, West Lafayette, IN 47907, USA (e-mail: djlove@ecn.purdue.edu).}

}}
\maketitle

\begin{spacing}{1.0}

\begin{abstract}
Communication at millimeter wave frequencies will be one of the essential new technologies in 5G.  Acquiring an accurate channel estimate is the key to facilitate advanced millimeter wave hybrid  multiple-input multiple-output (MIMO) precoding techniques. Millimeter wave MIMO channel estimation, however, suffers from a considerably increased channel use overhead. This happens due to the limited number of radio frequency (RF) chains that prevent the digital baseband from directly accessing the signal at each antenna. To address this issue, recent research has focused on adaptive closed-loop and two-way channel estimation techniques.  In this paper, unlike the prior approaches, we study a non-adaptive, hence rather simple, open-loop millimeter wave  MIMO channel estimation technique. We present a simple random design of channel subspace sampling signals and show that they obey the restricted isometry property (RIP) with high probability. We then formulate the channel estimation as a low-rank subspace decomposition problem and, based on the RIP, show that the proposed framework reveals resilience to a low signal-to-noise ratio. It is revealed that the required number of channel uses ensuring a bounded estimation error is linearly proportional to the degrees of freedom of the channel, whereas it converges to a constant value if the number of RF chains can grow proportionally to the channel dimension while keeping the channel rank fixed.  In particular, we show that the tighter the RIP characterization the lower the channel estimation error is. We also devise an iterative technique that effectively finds a suboptimal but stationary solution to the formulated problem. The proposed technique is shown to have improved channel estimation accuracy with a low channel use overhead as compared to that of previous closed-loop and two-way adaptation techniques.
\end{abstract}

\end{spacing}

\begin{spacing}{1.64}
\section{Introduction} \label{section1}

With the growth in the demand for wireless broadband, the operating frequency of modern wireless systems is steadily shifting upward to the millimeter wave band to provide a much wider bandwidth \cite{rappaport}.
In the millimeter wave bands, the radio channel experiences severe pathloss that is compensated for by using large-scale antenna arrays \cite{Torkildson2011, Heath16, Hur13}.
To make the system available at low cost, large-sized analog arrays that are fed by a limited number of radio frequency (RF) chains are popularly discussed.
This system is often referred to as a millimeter wave hybrid multiple-input multiple-output (MIMO) system.
Despite the limited RF chains, a remarkable throughput boost is possible if the system employs advanced hybrid MIMO precoding techniques \cite{roh2014mill,Hur13,Ayach14,hadi2015,alk}.
To enable such techniques, accurate channel state information (CSI) is the key.

Useful CSI has to be attained via channel estimation.
However, millimeter wave hybrid MIMO channel estimation is limited by an insufficient number of RF chains.
This limitation increases the channel use overhead required to extract a useful channel estimate of the system.
To address this difficulty, recent research has focused on adaptive channel estimation frameworks, including adaptive closed-loop \cite{Hur13,alk,Wang09,song2015adaptive,Noh17} and two-way \cite{hadi2015} channel estimation methods.
Closed-loop channel estimation includes techniques such as the multi-level adaptive subspace search techniques \cite{Hur13,alk,Wang09,song2015adaptive,Noh17}, where the receiver at each search level conveys a feedback to the transmitter to guide the next level subspace sampling signal design \cite{Duly,choi2014downlink,Noh14,Noh16}.
The two-way technique \cite{hadi2015} exploits the uplink-downlink channel reciprocity to iteratively refine the channel estimate via uplink-downlink channel echoing \cite{dahl2004blind,tang2005iterative,withers2008echo,Ogbe17}.
Improved estimation accuracy compared to other non-adaptive techniques, e.g., open-loop techniques, has been evidenced in \cite{Hur13,alk,Wang09,song2015adaptive,Noh17,hadi2015}.

Nevertheless, these processes still require a relatively large number of channel uses because of the repeated channel sampling \cite{Hur13,alk,Wang09,song2015adaptive,Noh17,hadi2015}.
Achieving lower channel use overhead is particularly important to the millimeter wave communications that are limited by a shortage of channel coherence resources.
Open-loop MIMO channel estimation has been thoroughly studied in the literature  \cite{hassibi2003much, Baum11}.
It is non-adaptive and simple.
However, it is impractical for estimating large-dimensional channels, such as sub-6 GHz massive MIMO channels  \cite{choi2014downlink,Duly,Noh14}, due to its inaccuracy.
Meanwhile, there have been convincing experimental studies showing that the practical millimeter wave MIMO channel exhibits substantial rank sparsity \cite{5gWhite,hur2016proposal}.
A general belief is that the sparsity of a signal, either in its support or subspace, provides tremendous clues when estimating it from distorted observations.
Useful directions for open-loop millimeter wave MIMO channel estimation are recently suggested by the low-rank matrix reconstruction techniques \cite{davenport2016overview,recht, chen2015exact, Jain10, cai2013compressed,  cai2015rop, candes2011tight, koren2009matrix, Haldar_factor}.

Methods for low-rank matrix reconstruction have recently spurred considerable interest due to their exploitation of rank sparsity \cite{recht, chen2015exact, Jain10, cai2013compressed,  cai2015rop, candes2011tight, koren2009matrix, Haldar_factor}.
The reliability of low-rank matrix reconstruction is best described by the restricted isometry property (RIP) \cite{recht, candes2011tight}.
The RIP models, in probability, the convergence of the subspace-sampled matrix to its original matrix.
It has also been used to analyze the estimation error performances of various algorithms \cite{jain2013low,cai2013compressed,candes2011tight}.
The RIP concept is useful and can be exploited in devising new design criteria for the open-loop channel estimation.
However, low-rank matrix reconstruction techniques \cite{recht, chen2015exact, Jain10, cai2013compressed,  cai2015rop, candes2011tight, koren2009matrix, Haldar_factor}, in their current form, face difficulties in estimating the millimeter wave MIMO channels.
Reliable performance of low-rank matrix reconstruction techniques \cite{recht, chen2015exact, Jain10, cai2013compressed,  cai2015rop, candes2011tight, koren2009matrix, Haldar_factor} is only ensured in the no noise or asymptotically high signal-to-noise ratio (SNR) regime.
Reliable operation at low SNR is critical for millimeter wave systems that are limited by heavy mixed signal processing with an excessive power consumption \cite{Heath16}.
Moreover, the probabilistic analysis of when the RIP holds has been studied for the subspace samples with independently and identically distributed (i.i.d.) Gaussian or Bernoulli entries \cite{recht,candes2011tight}.
Yet, the subspace samples that the millimeter wave arrays generate are not well modeled by these distributions.
Thus far, its RIP characterization has not been reported in the literature.

In this paper, we study an open-loop channel estimation technique by leveraging low-rank matrix reconstruction for millimeter wave hybrid MIMO systems.
This is in contrast with prior approaches like \cite{Hur13,alk,Wang09,song2015adaptive,Noh17,hadi2015}, that focus on closed-loop and two-way techniques.
We begin by introducing a random subspace sampling signal design, which does not rely on any prior knowledge of channel subspace information.
We adopt the precoder and combiner gain maximization criterion \cite{love2008overview} to formulate the channel estimation as a low-rank subspace decomposition problem.
We show through an analysis that the random ensembles of the subspace sampling signals satisfy the RIP with high probability.
On the basis of the established RIP, we carry out channel estimation error analysis and show that the proposed low-rank subspace decomposition reveals resilience to low SNRs.

In particular, we reveal that a tighter RIP bound characterization results in a lower estimation error for the devised low-rank subspace decomposition problem.
It is also shown that the required number of channel uses ensuring a bounded estimation error performance scales in proportional to the degrees of freedom of the channel when the number of RF chains is fixed, whereas it can converge to a constant value if the number of RF chains grows proportionally to the channel dimension while keeping the channel rank fixed.
It is worthwhile to note that the proposed technique does not assume any explicit channel model and statistics, making it useful in a wide range of sparse scenarios.
The subspace decomposition is a non-convex operation.
To effectively solve the problem, we devise an alternating optimization technique that finds a suboptimal but stationary solution to the problem.
The proposed technique results in improved channel estimation accuracy compared with prior approaches while consuming a considerably reduced channel use overhead.

This result is interesting since it is achieved without advanced adaptation techniques such as closed-loop or two-way adaptations  \cite{Hur13,alk,Wang09,song2015adaptive,Noh17,hadi2015}.
We are not aware of similar work establishing the RIP of a specific wireless signaling and estimation problem such as the low-rank millimeter wave MIMO channel estimation problem by drawing a connection between the RIP and estimation performance.
Though we focus exclusively on millimeter wave MIMO channel estimation in this work, we hope that our approach will motivate the use of the RIP in other related communication problems.

The rest of the paper is organized as follows.
Section II provides the system model and a brief review of the low-rank matrix reconstruction.
Section III presents the low-rank subspace decomposition  problem and the subspace sampling signal design.
In Section IV, the RIP of the subspace sampling signals is analyzed.
Section V presents the performance analysis and the algorithm to effectively solve the low-rank subspace decomposition  problem.
Finally, the numerical simulations and concluding remarks are presented in Sections VI and VII, respectively.

\subsubsection*{Notations}
A bold lower case letter $\ba$ is a vector, a bold capital letter $\bA$ is a matrix.
$\bA^T,\bA^H,\bA^{\!-1}$, tr($\bA$), $| \bA  |$, $\| \bA \|_F$, $\| \bA  \|_*$, and $\|\ba\|_2$ are, respectively,    the transpose,
conjuagate transpose, inverse, trace, determinant, Frobenius norm, nuclear norm (i.e., the sum of the singular values of $\bA$) of $\bA$, and $l_2$-norm of $\ba$.
$[\bA]_{:,i}$, $[\bA]_{i,:}$, $[\bA]_{i,j}$, and  $[\ba]_i$ are, respectively, the $i$th column, $i$th row, $i$th row and $j$th column entry of $\bA$, and $i$th entry of vector $\ba$.
$\vec(\bA)$ stacks the columns of $\bA$ and form a long column vector.
$\diag(\bA)$ extracts the diagonal entries of $\bA$ to form a column vector.
$\bI_M \! \in \! \R^{M\times M}$ is the identity matrix.
$\bA  \otimes \bB$ is the Kronecker product of $\bA$ and $\bB$.

\section{System Model} \label{section2}
We provide the signal model and briefly review the low-rank matrix reconstruction techniques under consideration.

\subsection{Signal Model}
\begin{figure}[t]
	\centering
	\includegraphics[width=5.2in, height=2in]{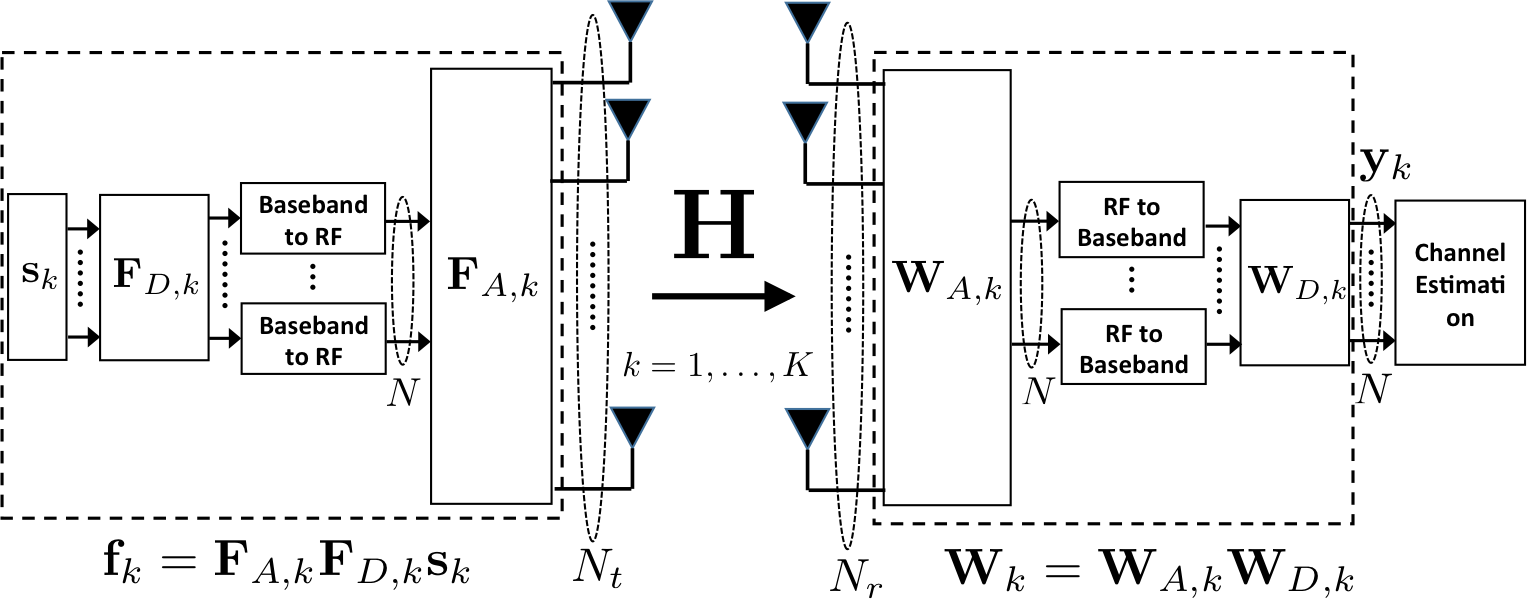}
	\caption{Open-loop millimeter wave hybrid MIMO channel sounding.} \label{Diagram}
\end{figure}

Consider a point-to-point millimeter wave hybrid MIMO system.
An independent block fading channel with a coherence block length $T_c$ (channel uses) is assumed.
The open-loop channel estimation is based on sounding a set of predefined channel subspace sampling signals from the transmitter to the receiver by using $K$ ($K \ll T_c$) channel uses.
We denote by $N_r$, $N_t$, and $N$, respectively, the numbers of receive antennas, transmit antennas, and RF chains ($N \!\ll\! \min(N_r,N_t)$).
{Fig. \ref{Diagram} illustrates the open-loop channel sounding at the $k$th channel use}, $k=1, \ldots, K$.
The received signal $\by_k$ is described as
\beq
\by_k = \bW_k^H \bH \bff_k + \bW_k^H \bn_k,  \label{channel sounding}
\eeq
where $\bW_k \!\in\! \C^{N_r \times N}$, $\bH \!\in\! \C^{N_r \times N_t}$, $\bff_k \!\in\! \C^{N_t \times 1}$, and $\bn_k \!\in\! \C^{N_r \times 1}$ are, respectively,  the receive subspace sampling signal, channel, transmit subspace sampling signal, and the noise.

{Seen from Fig. \ref{Diagram}}, the subspace sampling signals $\bW_k$ and $\bff_k$ in \eqref{channel sounding} are further divided into the analog and digital components, i.e.,
\beq
\bW_k=\bW_{A,k}\bW_{D,k}~ \text{and} ~\bff_k=\bF_{A,k} \bF_{D,k}\bs_k, \label{hybrid sounder0}
\eeq
where $\bW_{A,k} \!\in\! \C^{N_r \times N}$ ($\bF_{A,k} \!\in\! \C^{N_t \times N }$) and $\bW_{D,k} \!\in\! \C^{N \times N}$ ($
\bF_{D,k}\!\in\!\C^{N \times N}$) are the analog and digital receive (transmit) sampling signals, respectively.
$\bs_k\!\in\!\C^{N \times 1}$ is the transmitted symbol.
The entries of $\bF_{A,k}$ and $\bW_{A,k}$ obey the constant modulus property such that $| [\bF_{A,k}]_{m,n} | \!\=\! 1/\sqrt{N_t}$ and $|[\bW_{A,k}]_{m,n} | \!\=\! 1/\sqrt{N_r}$, $\forall m, n$, representing the analog phased-array elements.
We assume the entries of $\bn_k$ are independent and identically distributed (i.i.d.) Gaussian with zero mean and $\sigma^2$ variance.
The low-rank channel $\bH$ has $\rank(\bH)\=L$ with $L\leq N \ll \min(N_r,N_t)$.
With the power constraint $E [ \| \bff_k \|_2^2 ]\= 1$, the SNR per channel use is given by $\frac{1}{\sigma^2}$.

For the generation of the subspace sampling signal pair $(\bW_k, \bff_k)$,   we take a simple approach.
The transmitter only activates one RF chain per channel use for channel sounding.
This is to say the transmit subspace sampling signal is in the form $\bff_k \= \bff_{A,k}  f_{D,k} s_k$ where $\bff_{A,k}\!\in\!\C^{N_t \times 1}$, $\redd{f_{D,k}\!\in\!\C}$, and $s_k\!\in\!\C$ for simplicity.
Moreover, we admit the following random generation,
\beq
\d4\d4&&\d4{[\bW_{A,k}]}_{m,n} \= \frac{1}{\sqrt{N_r}} e^{ j\eta_{k,m,n}},
{[\bff_{A,k}]}_{n} \=\frac{1}{\sqrt{N_t}} e^{j\phi_{k,n}},\nonumber \\
\d4\d4&&\d4 \hspace{0.5cm} \bW_{\mathrm{D},k}= \bI_{N},~ f_{D,k} = 1,~\text{and}~ s_{k}=1,~ \forall m,n, \label{hybrid sounder}
\eeq
where random variables $\{ \eta_{k,m,n} \}$ and $\{ \phi_{k,n} \}$ are mutually independent and i.i.d. according to the uniform distribution over $[ 0,2\pi)$, i.e., $\eta_{k,m,n}, \phi_{k,n} \in \cU[0, 2\pi)$, $\forall k,m,n$.
The $\bff_k\= \bff_{A,k}  f_{D,k} s_k$ in \eqref{hybrid sounder} satisfies $E [ \| \bff_k \|_2^2 ]\= 1$.
Although simple, this random design, as we will show in the following sections, permits theoretical insights to conduct channel estimation performance analysis based on the RIP.

The digital baseband signal in \eqref{channel sounding} cannot directly access $\bH$, but it instead observes a noisy version of the analog distorted effective channel $\bW_k^H \bH \bff_k\! \in\! \C^{N \times 1}$ per channel use.
Hence, the channel sounding has to rely on the subspace sampling signal set $\cS \=\{(\bW_k ,\bff_k )\}_{k=1}^K$.
In general, observing the entire $\bH$ using $\cS$ requires $K\=O(N_rN_t)$ channel uses.
This overhead could overwhelm the channel coherence time $T_c$ in millimeter wave bands, which is over ten-fold shorter than that in sub-$6$ GHz bands.
On the other hand, rank sparsity of the channel matrix is an inherent characteristic of radio channels in the millimeter wave bands \cite{5gWhite,hur2016proposal}.


\subsection{Low-Rank Matrix Reconstruction }
Low-rank matrix reconstruction has received much recent attention in the field of large-scale sparse signal processing \cite{davenport2016overview,recht, chen2015exact, Jain10, cai2013compressed,  cai2015rop, candes2011tight, koren2009matrix, Haldar_factor}.
In its problem formulation, distorted observations are modeled by using an affine map $\cA$.
For the sounding model in \eqref{channel sounding}, we can define the affine map $\cA_{\cS}\!:\! \C^{N_r\times N_t} \!\!\mapsto \! \C^{KN\times 1}$ of $\bH$ associated with the subspace sampling signal set $\cS =\{ ( \bW_k, \bff_k )  \}_{k=1}^{K}$ as
\beq
\begin{aligned}
\textstyle
\cA_{\cS}(\bH) & \= \!
\left[ (\bW_1^H \bH \bff_1\!)^T \!,  \ldots , (\bW_K^H \bH \bff_K\!)^T   \right]^T  \in   \C^{K\! N  \times  1}. \label{affine map}
\end{aligned}
\eeq
Stacking $\{ \by_k \}_{k=1}^K$ in \eqref{channel sounding} into a long column vector gives
\beq
\textstyle
\by = \cA_{\cS}(\bH) + \tbn, \label{affine sampling}
\eeq
where $\by \= [\by_1^T, \ldots, \by_K^T ]^T \in \C^{KN \times 1}$ and the effective noise
\beq
\tbn = [(\bW^H_1 \bn_1)^T, \ldots, (\bW^H_K \bn_K)^T]^T  \in \C^{KN \times 1}. \label{effective noise}
\eeq
We will use the notation $M\! \triangleq \! KN$ to indicate the total number of observations.
In low-rank matrix reconstruction, popularly adopted approaches include the nuclear norm minimization (NNM) \cite{recht} and matrix factorization (MF) \cite{koren2009matrix, Haldar_factor, jain2013low}.
These approaches can directly be extended to estimate $\bH$ in \eqref{affine sampling}.

\subsubsection{Nuclear Norm Minimization (NNM)}
The NNM problem in \cite{recht} is formulated as
\beq
\begin{aligned}
\widehat{\bH} = \argmin  \limits_ {\bH}  \lA \bH \rA_* ~~\text{subject to~~}  \by =  \cA_{\cS}(\bH). \label{nuclear norm}
\end{aligned}
\eeq
This problem is convex and can be solved reliably.
However, as seen from \eqref{nuclear norm}, it does not guarantee the low-rank solution.
One critical drawback is that it does not scale to large-dimensional matrices due to prohibitively high computational complexity \cite{ZhangWei, recht}.

\subsubsection{Matrix Factorization (MF)} \label{subsubsection MF}
The development of MF \cite{koren2009matrix,Haldar_factor,jain2013low} has revealed a new path
towards handling large-scale low-rank matrices.
It achieves very accurate reconstruction for various classes of matrices at high SNR \cite{davenport2016overview}.
Given  \eqref{affine sampling}, the MF problem can be formulated as
\beq
\begin{aligned}
\big( \bL^\star, \bR^\star \big) = \argmin \limits_ {\bL\in \C^{N_r \times L},\bR\in \C^{N_t \times L}} \lA \by - \cA_{\cS}( \bL \bR^H) \rA_2^2.  \label{matrix factor}
\end{aligned}
\eeq
Solving \eqref{matrix factor} yields the estimate $\widehat{\bH} = \bL^\star(\bR^\star)^H$.
The problem in \eqref{matrix factor} is non-convex.
A widely adopted suboptimal treatment is to use the alternating minimization with power factorization  \cite{Haldar_factor}.
The complexity of the latter approach is much lower than that of NNM.
Interestingly, MF outperforms NNM when the SNR is high and the matrix is substantially low-rank \cite{Haldar_factor, jain2013low,ZhangWei}.
A critical drawback of MF is that it has very poor performance at low SNR \cite{ZhangWei}.
This is because the objective function in \eqref{matrix factor} is highly susceptible to noise.
For example, detrimental over-matching can frequently happen when the noise overpowers the signal, i.e., $\|\tbn\|_2 \gg \|\cA_{\cS }(\bH)\|_2 $ \cite{ZhangWei}.

\subsubsection{Restricted Isometry Property (RIP)}
A reliability guarantee of the low-rank matrix reconstruction is best described by the RIP \cite{recht,candes2011tight}.
For any $\bH$ whose rank is at most $L$ and an arbitrary affine map $\cA$, the RIP can be expressed as the bound,
\beq
| \|\cA(\bH)\|_2^2 - \|\bH \|_F^2 | \le \delta_L   \|\bH \|_F^2, \label{RIP bound}
\eeq
where $0 \le \delta_L <1$ is the $L$-RIP constant \cite{recht}.
The bound in \eqref{RIP bound} is justified when $\cA(\bH)$ is as significant as if $\bH$ were directly observed.
The reliable reconstruction of NNM \cite{recht, cai2013compressed} and MF \cite{jain2013low} is guaranteed with high probability when the $\cA(\bH)$  obeys the RIP.
The RIP also provides an operational characterization that any $\cA$ obeying \eqref{RIP bound} ensures a bounded estimation error  for  $K\=O(L \cdot \max(N_r, N_t))$  \cite{candes2011tight}.

\subsection{General Statement of Technique}
The desired open-loop millimeter wave MIMO channel estimation must leverage the sparse property of the channel while consuming far less channel use overhead than $K\=O(N_rN_t)$.
The RIP offers a useful statistical characterization that ensures a bounded estimation error performance while consuming a reduced amount of channel uses.
However, the RIP has only been modeled for the $\cS$ with i.i.d. Gaussian or Bernoulli entries \cite{recht}.
The $\cS$ that hybrid arrays generate based on  the random design in \eqref{hybrid sounder}, however, is in different random ensembles.
Moreover, the conventional low-rank matrix reconstructions are critically limited at low SNR.
For instance, the rank constraint in \eqref{matrix factor}, i.e., $\rank(\hat{\bH})\!\leq \! L$, is insufficient to capture the sparsity under the noise.
The original MF \cite{koren2009matrix,Haldar_factor} was originally devised for the noise-free scenario where the rank constraint was sufficient.
Note that with the noise, the matrix becomes full rank.
The desired open-loop millimeter wave MIMO channel estimation will require further stringent approaches to incorporate the sparsity under the noise.

\section{low-rank subspace decomposition } \label{SD section}
We begin by addressing the questions related to how to enhance the low SNR performance.
In  MIMO transmission, high resilience to noise is conventionally  achieved by matching the precoder and combiner to the channel subspace \cite{love2008overview}.
This points that when recovering a sparse matrix under the noise, constraints on the channel subspace would be more essential than the rank constraint in \eqref{matrix factor}.

\subsection{Problem Formulation}


Consider the unitary precoding and combining gain maximization problem,
\beq
\begin{aligned}
\!\! \left( \widehat{\bW}, \widehat{\bF} \right) \!\=&  \argmax\limits _{\bW, \bF} \! \lA  \bW^H \bH \bF  \rA_F^2 \ \!  \\
&\text{~subject to}~~~\bW^H \bW = \bI_d,~\bF^H \bF = \bI_d, \label{precoding gain}
\end{aligned}
\eeq
where $\bW \in \C^{N_\mathrm{r} \times d}$ and $\bF \in \C^{N_\mathrm{t} \times d}$ are semi-unitary matrices, and $d ~\leq L$ denotes the dimension of $\bW$ and $\bF$.
Equation \eqref{precoding gain} can be written in an alternative form \cite{ZhangWei}
\beq
\begin{aligned}
\!\! \big( \widehat{\bU}, \widehat{\bSig}, \widehat{\bV} \big) \!\ =&\argmin_{\bU,  \bSig,\bV }  \lA \bH - \bU \bSig \bV^H  \rA_F^2 \\
&\text{~subject to}~~~ ~~~\bU^H \bU = \bI_d,~\bV^H \bV = \bI_d,
\label{eq_decom1}
\end{aligned} \label{channel recons}
\eeq
where $\bU \in \C^{N_r \times d}$ and $\bV \in \C^{N_t \times d}$. In \eqref{channel recons}, $\bSig \in \C^{d \times d}$ is a diagonal matrix.
The relation between \eqref{precoding gain} and \eqref{channel recons} is that the optimal solution $\widehat{\bU}$ and $\widehat{\bV}$ in \eqref{channel recons} will also be optimal for \eqref{precoding gain}.
However, $\widehat{\bW}$ and $\widehat{\bF}$ in \eqref{precoding gain} are not necessarily optimal for \eqref{channel recons}.
An important implication of \eqref{precoding gain} and \eqref{channel recons} is that maximizing the precoding gain is equivalent to approximating $\bH$ to a subspace decomposition $\bU\bSig\bV^H$.

Sampling $\bH$ using $\cS \= \{(\bW_k, \bff_k)\}^K_{k=1} $, incorporating the affine map $\cA_{\cS }$ in \eqref{affine map}, and including the effective noise in \eqref{effective noise} to \eqref{channel recons} yields the low-rank subspace decomposition  problem,
\beq
\begin{aligned}
\!\! \big( \widehat{\bU}, \widehat{\bSig}, \widehat{\bV} \big) \!\ =&\argmin_{\bU,  \bSig,\bV }  \lA \by  -
\cA_{\cS}(\bU \bSig \bV^H)   \rA^2_2  \\
&\text{~subject to}~~~ ~~~\bU^H \bU = \bI_d,~\bV^H \bV = \bI_d. \label{eq_objective}
\end{aligned}
\eeq
The significance of \eqref{eq_objective} compared to \eqref{matrix factor} is that now \eqref{eq_objective} explicitly contains the subspace constraints and $\bSig$ can be exploited to improve robustness at low SNR.

When the noise overpowers the signal, the objective in \eqref{eq_objective} is likely misled by $\| \tilde{\bn} \-\cA_{\cS}(\bU\bSig\bV^H)  \|_2^2$.
This results in the critical over-matching problem \cite{ZhangWei}.
As a remedy, we  propose to control $\bSig$ to prevent over-matching by adding a condition to \eqref{eq_objective}, $\| \bSig \|_F^2 \!\leq\! \beta$, with $\beta$ being a constant, yielding
\beq
\begin{aligned}
	&\!\! \big( \widehat{\bU}, \widehat{\bSig}, \widehat{\bV} \big) \!\ =\argmin_{\bU,  \bSig,\bV }  \lA \by  -
	\cA_{\cS}(\bU \bSig \bV^H)   \rA^2_2  \\
	&\hspace{2.2cm} \text{~subject to}~~\bU^H \bU = \bI_d,~\bV^H \bV = \bI_d,  ~\|  \bSig  \|_F^2 \leq \beta. \label{eq_objective with P}
\end{aligned}
\eeq
The channel estimate is given by $\widehat{\bH}\= \widehat{\bU}\widehat{\bSig}\widehat{\bV}^H$.
For the selection of $\beta$, we first notice that the true channel $\bH$ is feasible to the problem in \eqref{eq_objective with P}.
This being said, choosing the smallest $\beta$ satisfying the condition {$ \| \bH \|_F^2 \! \leq \! \beta$} is the desired choice.
However, we will see in the simulation study in Section \ref{section6}, that a rough guess of $\beta$ is sufficient to provide robust estimation performance.
Notice that the subspace decomposition in \eqref{eq_objective with P} is non-convex due to the coupled variables and quadratic equality constraints.
An algorithm for effective minimization of \eqref{eq_objective with P} will be discussed in Section \ref{section SD}.

\subsection{Random Subspace Sampling Signal}
Based on the design in \eqref{hybrid sounder}, inserting \eqref{hybrid sounder} in \eqref{channel sounding} gives
\begin{eqnarray}
\by_k  
\= \bW_{A,k}^H \bH \bff_{A,k}  +  \bW_{A,k}^H \bn_k.  \label{hybrid channel sounding}
\end{eqnarray}
The $i$th entry of $\bW_{A,k}^H \bH \bff_{A,k} \in \C^{N\times 1}$ can be rewritten as
\beq
[ \bW_{A,k}^H \bH \bff_{A,k} ]_i = \tr(  \bff_{A,k} [\bW_{A,k}]_{:,i}^H \bH ) = \tr(\bX_{k,i}^H \bH), \label{entrywise}
\eeq
where $\bX_{k,i} = [\bW_{A,k}]_{:,i} {\bff}_{A,k}^H  \in \C^{N_r \times N_t}$, $1\le i \le N$.
Provided $M\=KN$ observations, the affine map $\cA_{\cS}(\bH): \C^{N_r \times N_t} \rightarrow \C^{M \times 1}$ in \eqref{affine map} can now be rewritten as
\beq
\begin{aligned}
\cA_{\cS}(\bH) =& [\tr(\bX_{1,1}^H\bH),\ldots, \tr(\bX_{1,N}^H\bH), \tr(\bX_{2,1}^H\bH), \ldots \\
                         & \hspace{1.5cm}  \ldots,\tr(\bX_{K,1}^H\bH), \ldots, \tr(\bX_{K,N}^H\bH)]^T.  \label{matrix form sampling1}
 \end{aligned}
\eeq

The following lemma characterizes the distribution of each entry of $\bX_{k,i}\in\C^{N_r\times N_t}$.

\begin{lemma} \label{entry distribution}
The entries of $\bX_{k,i}= [\bW_{A,k}]_{:,i} {\bff}_{A,k}^H  \in \C^{N_r \times N_t}$, $1\le i \le N$ in \eqref{matrix form sampling1}  are i.i.d. and each entry  follows
\beq
{[\bX_{k,i}]_{m,n}} \stackrel{d}{=} \frac{1}{\sqrt{N_r N_t}} e^{j\theta}, \label{phase rotation}
\eeq
 where $\stackrel{d}{=}$ denotes equality in distribution and $\theta \sim \cU[0,2\pi)$.
\end{lemma}

\begin{proof}
  From  \eqref{matrix form sampling1} and \eqref{hybrid sounder}, one can write
  \beq
  \begin{aligned}
   {[\bX_{k,i}]_{m,n}} &= [\bW_{A,k}]_{m,i} [\bff_{A,k}]_n^H \\
    & = \frac{1}{\sqrt{N_r N_t}} e^{j(\eta_{k,m,i}-\phi_{k,n})}, \label{random entry}
   \end{aligned}
  \eeq
  where $\eta_{k,m,i}$ and $\phi_{k,n}$ are independent and $\eta_{k,m,i}, \phi_{k,n} \!\in\! \cU[0, 2\pi)$.
  Let  $\theta_t \= \eta_{k,m,i}-\phi_{k,n}$.
  Then $\theta_t$ follows the probability distribution function (PDF)
  \beq
  f(\theta_t) =
  \begin{cases}
   \frac{1}{4\pi^2}(2\pi -\theta_t), ~~ \text{if} ~\theta_t \in [0,2\pi), \\
  \frac{1}{4\pi^2} (2\pi+ \theta_t),~~ \text{if}~ \theta_t \in [-2\pi,0),\\
    0,  ~~\text{otherwise.} \label{pdf phase}
  \end{cases}
  \eeq
  Since $e^{j\theta_t} \stackrel{d}{\=} e^{j\mod(\theta_t,2\pi)}$, where $\mod(\theta_t,2\pi)$ is the modulo $2\pi$ of $\theta_t$, defining $\theta \!\triangleq\! \mod(\theta_t,2\pi)$ yields
  \beq
  f(\theta) =
  \begin{cases}
   \frac{1}{2\pi}, ~~ \text{if} ~\theta \in [0,2\pi) \\
    0, ~~\text{otherwise} \label{pdf phase mod}.
  \end{cases}
  \eeq
  Since $\eta_{k,m,i}$ and $\phi_{k,n}$ are i.i.d., it is straightforward to conclude that ${[\bX_{k,i}]_{m,n}} $ is i.i.d. $\forall m,n$.
\end{proof}
The distribution in Lemma \ref{entry distribution} will be used in the next section to establish the RIP.

\section{Restricted Isometry Property} \label{section plus}

In this section, we show that the affine map $\cA_{\cS}$ in \eqref{matrix form sampling1} satisfies the RIP in \eqref{RIP bound} with high probability.

\subsection{Preliminaries}
The following theorem states a sufficient condition for the RIP of an arbitrary affine map $\cA$.

\begin{theorem}[\!\cite{candes2011tight}] \label{RIP paper theorem}
Let an affine map $\cA \! : \!\C^{N_r \times N_t}\! \mapsto\! \C^{M \times 1} $ be a random ensemble obeying the following condition: for any $\bH \!\in\! \C^{N_r \times N_t}$ and fixed $0\!<\!\al\!<\!1$,
\beq
\Pr (| \| \cA(\bH) \|_2^2 - \| \bH \|_F^2 | > \al \| \bH \|_F^2) \le C e^{-cM}, \label{paper theorem}
\eeq
for $C,c>0$ which depends on $\al$.
Then, if $M \ge L C (N_t+N_r+1)$, $\cA$ satisfies the RIP in \eqref{RIP bound} with probability exceeding $1-C e^{-qM}$ for a fixed constant $q >0$, where $q = c-\frac{\ln(36 \sqrt{2}/\del_L)}{C}$ and $\del_L$ is the $L$-RIP constant.
\end{theorem}

In \eqref{paper theorem}, the probability is taken over the affine map $\cA$.
Here, $\cA$ is not restricted to the sounder set $\cS$ in \eqref{matrix form sampling1} but it is modeled for any arbitrary $\cS$.
Reindexing $\bX_{k,i}$ in \eqref{matrix form sampling1} as $\bX_m$ with $m$ being a function of $(k,i)$ such that $m\=(k\-1)N \+ i$, $k\=1, \ldots, K$ and $i \= 1, \ldots, N$, we have
\beq
[\cA (\bH) ]_m \= \tr( \bX_m^H \bH ), ~~  m\=1,\ldots, M.  \label{affine operator entry-wise}
\eeq
It is tractable to transform the sufficient condition in \eqref{paper theorem} in a vector form
\beq
 \Pr (| \| \bA \!\vec(\bH)  \|_2^2 \- \|  \vec(\bH)  \|_2^2 | \! >\!  \d4&\al&\d4  \|  \vec(\bH)  \|_2^2 )  \le  C e^{-cM}, \nonumber \\
 \d4& &\d4 \label{paper theorem vec}
\eeq
where $\bA \!\in\! \C^{M \times N_r N_t}$ and
$$[\bA]_{m,:}\=\vec(\bX_m)^H, ~\forall m.$$
The condition in \eqref{paper theorem vec} specifies  that the random ensembles of $ \| \bA\vec(\bH) \|_2^2$ should be concentrated around $\| \vec(\bH) \|_2^2$ in order to meet the RIP.

One such example of $\bA$ is the matrix with i.i.d. Gaussian entries \cite{elementary}, i.e., $[\bA]_{m,n} \sim \cN(0,\frac{1}{M})$.
Another example is when $\bA$ has entries sampled from an i.i.d. symmetric Bernoulli distribution \cite{database}, i.e.,
\beq
[\bA]_{m,n} \in \{  1/\sqrt{M} , -1/\sqrt{M} \} ~~ \text{with probability ${1}/{2}$}. \label{bern}
\eeq
Specifically, in \cite{database}, it has been shown that the random construction in \eqref{bern} satisfies \eqref{paper theorem vec} with $C \= 2$ and $c \= \al^2 / 4 \- \al^3 / 6$.

Accoridng to Lemma \ref{entry distribution}, the matrix $\bA$ associated with  $\cS\=\{ \bW_k, \bff_k \}_{k=1}^{K}$, which we will denote as $\bA_{\cS}$, has i.i.d. entries following
\beq
[\bA_{\cS}]_{m,n} \stackrel{d}{=} \frac{1}{\sqrt{N_r N_t}} e^{j\theta} \label{unit entry}
\eeq
with $\theta \in \cU[0, 2\pi)$.
Each entry of \eqref{unit entry} can be viewed as a random phase rotation  of $[\bA]_{m,n}$ in \eqref{bern}.
In what follows, we will establish the RIP of $\bA_{\cS}$ in \eqref{unit entry} by showing \eqref{paper theorem vec} (i.e., \eqref{paper theorem} in Theorem \ref{RIP paper theorem}).

\subsection{RIP Proof}

The following lemma establishes a connection between the vector drawn from the distribution in \eqref{unit entry} and the vector consisting of Bernoulli random variables in \eqref{bern}.

\begin{lemma} \label{expectation lemma}
 Let $\bA_{\cS} \!\in\! \C^{M \times N_r N_t}$ and $\bB \!\in\! \R^{M \times N_r N_t}$ be random matrices with  i.i.d. entries $[\bA_{\cS}]_{m,n} \= {1}/{\sqrt{N_r N_t}}e^{j\theta_{m,n}}$,  $\theta_{m,n}  \!\in\! \cU[0,2\pi) $, and  $[\bB]_{m,n} \!\in\! \{ {-1}/{\sqrt{N_r N_t}}, {1}/{\sqrt{N_r N_t}}\}$ with equal probability, respectively.
 Then, for any fixed unit norm vector $\bx \!\in\! \C^{N_r N_t \times 1}$, $\| \bx \|_2^2\=1$,
 \beq
    E\left[  \left|[\bA_{\cS}]_{m,:}{\bx}\right|^{2u}\right] \!\leq\! E\left[ \left|[\bB]_{m,:}\bar{\bx}\right|^{2u}\right], \label{theorem2 in}
 \eeq
where $u$ is a non-negative integer, $\bar{\bx}$ is the element-wise absolute of $\bx$, and the expectations are taken over $[\bA_{\cS}]_{m,:}$ and $[\bB]_{m,:}$, respectively.
\end{lemma}

\begin{proof}
See Appendix \ref{appendix1}.
\end{proof}

Now, we show that $\bA_{\cS}$ in \eqref{unit entry} meets the bound \eqref{paper theorem}.
The following theorem is inspired by a similar theorem in \cite[Theorem 5.1]{database} and refines the results for $\bA_{\cS}$ in \eqref{unit entry}.

\begin{theorem} \label{unit sensing matrix}
Let $\bA_{\cS} \!\in\! \C^{M \times N_r N_t}$ be a random matrix with i.i.d. entries $[\bA_{\cS}]_{m,n} \=\frac{1}{\sqrt{N_r N_t}} e^{j\theta}$, $\forall m,n$, where $\theta \!\in\! \cU[0,2\pi)$. Then, for any fixed $\bx \!\in\! \C^{N_r N_t \times 1}$ with $\| \bx \|_2^2\=1$ and $0\!<\!\al \!<\! 1$, the following holds
\beq
\Pr \lp \la  \lA  \! \sqrt{{N_t N_r}/{M}}\bA_{\cS} \bx  \rA_2^2 -1 \ra \geq \al  \rp\! \le \!  2 e^{-\frac{M}{2}\! \big(\! \frac{\al^2}{2} - \frac{\al^3}{3} \!\big)} . \label{upper tail in}
\eeq
\end{theorem}

\begin{proof}
See Appendix \ref{appendix2}.
\end{proof}

Based on Theorem \ref{RIP paper theorem} and Theorem \ref{unit sensing matrix}, we provide our main theorem of this section as follows.
\begin{theorem} \label{RIP measurements pre}
Suppose the affine map $\cA_{\cS}$ in \eqref{matrix form sampling1} and $\hat{\cA}_{\cS}(\bH) = \sqrt{\frac{N_\mathrm{r} N_\mathrm{t}} {M} }{\cA}_{\cS}(\bH) $.
For arbitrary $\bH \in \C^{N_r \times N_t}$ and any $0< \al <1$, the following holds
\beq
 \Pr \big(\big | \| \hat{\cA}_{\cS}(\bH) \|_2^2 - \| \bH \|_F^2  \big | \ge \al \| \bH \|_F^2 \big)  \leq  2 e^{\-\frac{M}{2}\lp \frac{\al^2}{2}-\frac{\al^3}{3} \rp}. \label{probility condition}
\eeq
Moreover,   if $M \ge 2L(N_t+N_r+1)$, then $\hat{\cA}_{\cS}$ satisfies the RIP in \eqref{RIP bound} with probability exceeding $1-2e^{-qM}$ and isometry constant $\del_L$, where $q = \frac{\al^2}{4}-\frac{\al^3}{6}- \frac{\ln(36 \sqrt{2}/\del_L)}{2}$.
\end{theorem}

\textit{Proof:}
Rewrite $\hat{\cA}_{\cS}(\bH) \= \sqrt{\frac{N_r N_t} {M}}\bA_{\cS}\vec(\bH)$ where, by Lemma \ref{entry distribution}, each entry of $\bA_{\cS} \in \C^{M \times N_r N_t}$ is i.i.d. and follows the same distribution as $\frac{1}{\sqrt{N_r N_t}}\exp(j\theta)$ for $\theta \in \cU[0,2\pi)$.
Then, by Theorem \ref{unit sensing matrix},  substituting $ \bx \=  \vec(\bH)/\| \bH \|_F $ into \eqref{upper tail in}, the following holds
\beq
\d4\d4 && \d4 \Pr \!\bigg( \bigg| \bigg\|\! \frac{\hat{\cA}_{\cS}(\bH)}{\|\bH\|_F} \!\bigg\|_2^2 \!\!\!\- 1  \bigg| \!\ge\! \al \!\! \bigg)
\!\=\Pr \!\Big( \big| \! \| \hat{\cA}_{\cS}(\bH) \|_2^2 \!\-\! \| \bH \|_F^2  \big | \!\ge\! \al \| \bH \|_F^2 \!\Big) \nonumber\\
\d4\d4 && \d4 \hspace{4.0cm} \le 2e^{\-\frac{M}{2}\lp \frac{\al^2}{2}-\frac{\al^3}{3} \rp}, \label{proof mar}
\eeq
i.e., the sufficient condition in Theorem \ref{RIP paper theorem}.
Then, by Theorem \ref{RIP paper theorem}, if \eqref{proof mar} holds and
\beq
M \ge 2L(N_t+N_r+1), \label{minimal obser}
\eeq
$\hat{\cA}_{\cS}$ satisfies the RIP for $L$-RIP constant {$\del_L$} with probability exceeding $1-2 e^{-qM}$ and fixed constant $q = \frac{\al^2}{4}-\frac{\al^3}{6}-\frac{\ln(36\sqrt{2}/\del_L)}{2}$. This concludes the proof. \hfill \qed

\begin{remark}
Theorem \ref{RIP measurements pre} signifies that  the subspace samples in \eqref{matrix form sampling1} become as significant as directly observing $\bH$ when the number of channel uses $K$ grows as large as $2\frac{L}{N}(N_t+N_r+1)$.
For fixed rank $L$, as  $N_r, N_t$, and $N$ tend to infinity while keeping the ratio $\frac{N_r + N_t}{N}=\gam$ fixed, the channel  $\bH$ becomes rank-sparse and $K {\rightarrow} 2 L \gam$.
On the other hand,  for fixed $L$ and $N$, as $N_r, N_t$ grow large, the overhead $K$ is on the order of $O(\max(N_r, N_t))$.
The RIP will further be consolidated in Section \ref{section SD} by analyzing the estimation error performance of the low-rank subspace decomposition  problem in \eqref{eq_objective with P}.
\end{remark}

\subsection{Asymptotic Analysis } \label{section central}
To gain intuition on the bound in \eqref{probility condition}, we model the distribution of $\|\hat{\cA}_\cS(\bH) \|_2^2$ in the asymptotically large $M$, $N_r$, and $N_t$ regime.
The central limit theorem (CLT) \cite{probability2010} provides a tractable means for this.
From \eqref{matrix form sampling1}, we have
\beq
\begin{aligned}
{[\hat{\cA}_{\cS}(\bH)]}_{m} & = \sqrt{\frac{N_\mathrm{r} N_\mathrm{t}} {M} } \tr(\bX_m^H \bH), ~~ m=1,\ldots, M. \nn
\end{aligned}
\eeq
 As $N_r,N_t\!\rightarrow\! \infty$, by the CLT, $\frac{\sqrt{M}}{\| \bH \| _F } [\hat{\cA}_{\cS}(\bH)]_m$ converges to the Gaussian normal distribution, i.e., $\frac{\sqrt{M}}{\| \bH \| _F } [\hat{\cA}_{\cS}(\bH)]_m$ $\sim$ $\cC \cN(0,1)$, resulting in
\beq
\big[\hat{\cA}_{\cS}(\bH) \big]_m \sim \cC \cN(0,\| \bH \| _F^2/M), ~ \forall m. \label{CLT}
\eeq
This implies $\frac{M} {\| \bH \| _F^2}\| \hat{\cA}_{\cS}(\bH) \|_2^2 \sim \chi^2(M)$, where $\chi^2(M)$ represents the central chi-squared distribution with $M$ degrees of fredom.
As $M\rightarrow \infty$, the central chi-squared distribution behaves  as $\chi^2(M)\sim \cN(M,2M)$ \cite{kayest}.
Hence, the asymptotic distribution of $\|\hat{\cA}_\cS(\bH) \|_2^2$ follows
\beq
\| \hat{\cA}_{\cS}(\bH) \|_2^2 \sim  \cN(\| \bH \| _F^2, 2\| \bH \| _F^4/M). \label{asympotic analysis 1}
\eeq

Based on \eqref{asympotic analysis 1}, the l.h.s. of \eqref{probility condition} is upper bounded by
\beq
&&\!\!\!\!\!   \Pr \!\bigg( \bigg| \bigg\|\! \frac{\hat{\cA}_{\cS}(\bH)}{\|\bH\|_F} \!\bigg\|_2^2 \!\!\!\- 1  \bigg| \!\ge\! \al \!\! \bigg)  = \frac {2} {\sqrt{2\pi}}\int _{-\infty} ^{-\al\sqrt{M/2}}\!\!\! e^{-\frac{t^2}{2} } dt \nonumber\\
\!\!\!\!\!& & \hspace{3.9cm}   \le \ 2 e^{-\frac{M}{2}\frac{\al^2}{2}} , \label{central pro}
\eeq
where the last step follows from the Chernoff bound of the normal Gaussian random variable  \cite{probability}.

\begin{remark}
It is interesting to examine the two-sided tail bounds in \eqref{proof mar} and \eqref{central pro}.
The comparison can be made by looking at the two exponents.
It can be seen that $\frac{\al^2}{2}-\frac{\al^3}{3} \leq \frac{\al^2}{2}$ and $\frac{\al^2}{2}-\frac{\al^3}{3} \rightarrow \frac{\al^2}{2}$ as $\al \rightarrow 0$, meaning for asymptotically large $M$, $N_r$, and $N_t$, the bound in \eqref{proof mar} converges to the CLT bound in \eqref{central pro} as $\alpha$ tends to zero.
This is to say, in the asymptotic sense, the RIP characterization is equivalent to the CLT whose convergence is in probability.
A fundamental difference is that the RIP in Theorem \ref{RIP measurements pre} gives a tighter characterization than the convergence in probability since it does not require that $M$, $N_r$, and $N_t$ tend to infinity, while the bound in \eqref{central pro} only holds for asymptotically large $M$, $N_r$, and $N_t$.

\end{remark}

\section{Estimation Error Analysis and An Algorithm for low-rank subspace decomposition } \label{section SD}
On the basis of the RIP, the estimation error performance of the low-rank subspace decomposition  problem in \eqref{eq_objective with P} is analyzed.
We also provide an algorithm for effective minimization of the low-rank subspace decomposition  problem.

\subsection{Estimation Error Bound}
Denote the compact sigular value decomposition (SVD) of the channel as $ \bH \= \bPsi \bLam \bPhi^H $, where $\bPsi\!\in\!\C^{N_r\times L}$ and $\bPhi\!\in\!\C^{N_t\times L}$ are semi-unitary matrices, i.e., $\bPsi^H \bPsi=\bI_L$ and $\bPhi^H\bPhi=\bI_L$, and $\bLam \in \R^{ L \times L}$ is the diagonal matrix with singular values $\lam_1 \geq \cdots \geq \lam_L$  at position $(l,l)$, $l=1, \ldots, L$.
Define the desired channel estimate of the subspace decomposition problem in \eqref{eq_objective with P} as $\bH_d$, where $\bH_d \= \sum_{l=1}^d \lam_l  \bpsi_l \bphi^H_l$ with  $\bpsi_l$ and $\bphi_l$ being the $l$th columns of $\bPsi$ and $\bPhi$ and $d\! \leq \! L$.
We also define $\bH_{L \setminus d} \= \bH - \bH_d \= \sum_{l=d+1}^L \lam_l  \bpsi_l \bphi^H_l$ as the residual.
For the squared error $ \big\| \bH_d - \widehat{\bH} \big\|_F^2 $ of the low-rank subspace decomposition  in \eqref{eq_objective with P}, we have the following theorem.
\begin{theorem} \label{the4}
Consider the problem in \eqref{eq_objective with P} where  $\bH \=  \bH_d \+ \bH_{L \setminus d}$ and $\cA_{\cS}$ is generated according to \eqref{matrix form sampling1}.
Given that $\sqrt{\frac{N_r N_t}{M}}\cA_{\cS}(\cdot)$ meets the RIP with $2d$-RIP constant $\del_{2d}$ for $M \ge 4d(N_t+N_r+1)$, we have
\redd{\beq
 \lA \bH_d - \widehat{\bH} \rA_F^2  \! \leq \! \min \!\left \{\frac{4 N_t N_r  (\lA \cA_{\cS}(\bH_{L \setminus d}) +  \tilde{\bn} \rA_2^2)} {(1-\delta_{2d}) M}  , 2\beta \right\}\!,\label{error bound}
\eeq}
where the bound \eqref{error bound} holds with probability exceeding $1\-2\exp(-qM)$ for $q \= \frac{\al^2}{4}\-\frac{\al^3}{6}\-\frac{\ln(36 \sqrt{2}/\del_{2d})}{2}$ and $0\!<\!\alpha\!<\!1$.
\end{theorem}
\begin{proof}
See Appendix \ref{appendix3}
\end{proof}

\begin{remark}
In \eqref{error bound}, as $M$ grows, \redd{$\| \cA_{\cS}(\bH_{L \setminus d}) +  \tilde{\bn} \|_2^2/M \rightarrow E\left[  \left| [\cA_{\cS}(\bH_{L \setminus d}) +  \tilde{\bn}]_i \right|^2 \right], i =1,2,\cdots,M$,} almost surely.
Meanwhile, the constant $\delta_{2d}$ in \eqref{error bound} can decrease as $M$ increases.
This can be explained as follows.
According to Theorem \ref{the4}, the bound \eqref{error bound} is guaranteed with  probability exceeding $1-2 e^{-q M}\triangleq \mu(\delta_{2d}, M)$ where $q \= \frac{\al^2}{4}\-\frac{\al^3}{6}\-\frac{\ln(36 \sqrt{2}/\del_{2d})}{2}$.
For $M_1\geq M_2 \geq \redd{4d(N_t+N_r+1)}$, in order to have the same level of probability guarantee, i.e., $\mu(\delta_{2d,1}, M_1)= \mu(\delta_{2d,2}, M_2)$, it requires $\delta_{2d,1} \leq \delta_{2d,2}$ because the exponent $q$ is proportional to $\delta_{2d}$.
This means that as $M$ grows the upper bound in \eqref{error bound}, which is proportional to $\delta_{2d}$, decreases.
\end{remark}

\begin{remark}
The fact that  the bound \eqref{error bound} is proportional to $\delta_{2d}$ suggests that optimizing the subspace sampling signal set $\cS$ by minimizing the RIP bound in \eqref{RIP bound}  would result in better execution of the channel estimation because the tighter the RIP bound, the lower $\delta_{2d}$ will be.
In this work, we do not aim to optimize $\cS$, but designing $\cS$ to minimize the RIP constant in \eqref{RIP bound} (i.e., reducing the error bound in \eqref{error bound}) is an interesting future research problem.
\end{remark}

\begin{remark}
When $L\=d$, $\bH_{L\setminus d} \=0$ in \eqref{error bound}, and this results in the least error bound
\beq \big\| \bH_d - \widehat{\bH} \big\|_F^2  \le  \min \left \{\frac{4 N_t N_r \|  \tilde{\bn} \|_2^2} {(1-\delta_{2d}) M}   , 2\beta \right\}. \label{seb} \eeq
Now, taking the expectation yields the mean squared error (MSE) bound
\beq E  \big[  \big\| \bH_d - \widehat{\bH} \big\|_F^2 \big]  \le  \min \left \{\frac{4 N_t N_r \sig^2} {(1-\delta_{2d})}   , 2\beta \right\}, \label{mmseb} \eeq
where we use the fact that $E [ \|  \tilde{\bn} \|_2^2 ]\=M\sig^2$.
It is noteworthy that, in the low SNR regime (i.e., when $\sigma^2$ is large), the MSE bounded by $2\beta$,  preventing the estimation error from growing along with the noise power $\sig^2$.
This can sharpen the MSE performance at  low SNR.
Overall, Theorem \ref{the4} formulates a required number of channel uses,  guaranteeing a bounded estimation error performance with high probability.
\end{remark}

\subsection{Alternating Optimization} \label{section4B}
It is difficult to directly find the solution to \eqref{eq_objective with P} due to the semi-unitary constraints, which are non-convex.
We turn to a suboptimal approach that effectively solves the problem.
To this end, we first consider the convex relaxation of \eqref{eq_objective with P} as
\beq
\begin{aligned}
&\hspace{0.5cm} \big( \widehat{\bU},\widehat{\bV},\widehat{\bSig} \big) = \argmin \limits_ {\bU,\bSig,\bV} \lA \by - \cA_{\cS}(\bU \bSig \bV^H) \rA_2^2
\\
&\hspace{3cm} \text{subject to~} \tr(\bU^H \bU)  \le d,\tr(\bV^H \bV )  \le d, \lA    \bSig   \rA _F^2 \le \beta,
\label{SVD method_relax at}
\end{aligned}
\eeq
where $\bSig \in \C^{d \times d}$ is the same diagonal matrix as in \eqref{eq_objective with P}.
The equality constraints are relaxed and the feasible set of \eqref{SVD method_relax at} becomes convex.
However,  the problem  is still non-convex since the optimization parameters are coupled with each other.
The coupling among $\bU$, $\bV$, and $\bSig$ in the objective function makes block coordinate decent \cite{Razaviyayn13} an ideal approach for \eqref{SVD method_relax at} to iteratively optimize $\bU$, $\bV$, and $\bSig$.
In particular, optimizing one parameter by fixing the other two parameters in \eqref{SVD method_relax  at} is convex.
Denoting $\ell$ as the interation index, we can iteratively optimize $\bU_{\ell}$, $\bV_{\ell}$, and $\bSig_{\ell}$ for $\ell=1, 2, \ldots$ by solving the following subproblems:
\begin{itemize}
\item[(S1)] Fix the row subspace $\bV_\ell$ and power allocation $\bSig_\ell$, optimize the column subspace $\bU_\ell$;
\item[(S2)] Fix the column subspace $\bU_\ell$ and power allocation $\bSig_\ell$, optimize the row subspace $\bV_\ell$;
\item[(S3)] Fix the row subspace $\bV_\ell$ and column subspace $\bU_\ell$, optimize the power allocation $\bSig_\ell$.
\end{itemize}

\noindent A formal description of the alternating optimization is provided in Algorithm \ref{alg1}.
In Step \ref{step2} of Algorithm \ref{alg1}, $\cA^*_{\cS}: \C^{M\times 1} \! \mapsto\! \C^{N_r \times N_t}$ denotes the adjoint operator of the affine map  $\cA_{\cS}: \C^{N_r\times N_t} \!\mapsto\! \C^{M\times 1}$.
Specifically, based on \eqref{affine operator entry-wise}, it follows $\cA^*_{\cS} (\by) \= \sum_{m=1}^{M} [\by]_m \bX_m$.
Hence, the initialization $\widehat{\bH}_{(0)} \= \cA^*_{\cS} (\by) = \sum_{m=1}^{M} [\by]_m \bX_m$ in Step \ref{step2}  can be viewed as  to linearly estimate the channel $\bH$ by taking the sounder set $\cS\=\{(\bW_k, \bff_k)\}_{k=1}^K$ as its basis and the noisy subspace samples $[\by]_m$  as their combining weights.

\begin{algorithm}[t]
	\caption{Alternating Minimization for Low-Rank Subspace Decomposition}
	\label{alg1}
	\begin{algorithmic} [1]
		\STATE Input: Affine map $\cA_{\cS}$, observations $\by$, and maximum number of iteration $\ell_{max}$.
		\STATE Initialization: Set $\ell=0$, initialize $\widehat{\bH}_{(0)} = \cA^*_{\cS}(\by)$, let $\bU_{(0)}$, $\bV_{(0)}$, and $\bSig_{(0)}$ be the top $d$ left singular vectors, right singular vectors, and singular value matirx of $\widehat{\bH}_{(0)}$, respectively. \label{step2}
		\REPEAT
		\STATE Update the column subspace ${\bU}_{(\ell)}$ to $\bU_{(\ell+1)}$ by solving
		\beq
		\text{(S1)}
		\begin{cases}
			\bU_{(\ell+1)} \! = \! & \d4 \argmin \limits_ {\bU} \big\| \by  \!-\! \cA_{\cS}(\bU \bSig_{(\ell)} \bV_{(\ell)}^H) \big\|_2^2\\
			&\! \! \! \! \! \! \text{subject to~}  \tr(\bU^H \bU) \le d,
			\nonumber
		\end{cases}
		\eeq
		\STATE Update the row subspace $\bV_{(\ell)}$ to $\bV_{(\ell+1)}$ by solving
		\beq
		\text{(S2)}
		\begin{cases}
			\bV_{(\ell+1)} \! = \! &\! \! \! \! \! \! \argmin \limits_ {\bV} \lA \by \!-\! \cA_{\cS}(\bU_{(\ell+1)} \bSig_{(\ell)} \bV^H) \rA_2^2\\
			&\! \! \! \! \! \! \text{subject to~}  \tr(\bV^H \bV) \le d,
			\nonumber
		\end{cases}
		\eeq
		\STATE Update the subspace power allocation matrix $\bSig_{(\ell)}$ to $\bSig_{(\ell+1)}$ by solving
		\beq
		\text{(S3)}
		\begin{cases}
			\bSig_{(\ell+1)}  \! = \! \argmin \limits_ {\bSig}  \! \big \| \by  \!-\! \cA_{\cS}(\bU_{(\ell+1)} \bSig \bV_{(\ell+1)}^H) \big\|_2^2,\\
			\text{subject to~}  \lA   \bSig  \rA _F^2 \le \beta, \nonumber
		\end{cases}
		\eeq
		\STATE${\widehat{\bH}}_{(\ell+1)} =\bU_{(\ell+1)} \bSig_{(\ell+1)} \bV_{(\ell+1)}^H$ and $\ell  =  \ell + 1$, \label{update}
		\UNTIL{$\ell$ exceeds $\ell_{max}$ or the iterations stagnate.}
		\STATE Output: $\widehat{\bH} =\widehat{\bH}_{(\ell)}$.
	\end{algorithmic}
\end{algorithm}

Since (S1), (S2), and (S3) are all convex, we can obtain the optimal solution of each subproblem in closed-form.
We limit our discussion to the column subspace optimization (S1) and the power matrix optimization (S3),
keeping in mind that the same approach for solving (S1) applies to the row subspace optimization (S2).
For simplicity, we omit the iteration index $\ell$ attached to the variables.
The following lemma provides the solution to (S1).

%
%

\begin{lemma}\label{t1}
Consider the following quadratic program,
\beq
\! \widehat{\bU} \= \argmin \limits_ {\bU} \! \lA \by \- \cA_{\cS}(\bU \bSig \bV^H) \rA_2^2 \nonumber\\
~ \text{subject~to~}  \tr(\bU^H \bU) \!\leq d,
\label{SVD method U omit}
\eeq
where $\by\in\C^{M\times 1}$, $\bU\in \C^{N_\mathrm{r} \times d}$, $\bV\in \C^{N_\mathrm{t} \times d}$, $\bSig \in \C^{d \times d}$, and $\cS\=\{(\bW_k, \bff_k)\}_{k=1}^K$.
Let $\bG\in \C^{M \times (d N_\mathrm{r})}$ be
\beq
\bG \! = \! [((\bSig \bV^H \bff_1)^T \! \otimes \! \bW_1^H)^T \!,\ldots,((\bSig \bV^H \bff_K)^T \! \otimes \!  \bW_K^H)^T]^T \!. \label{matrix B}
\eeq
Then, the optimal solution $\widehat{\bU}$ is given by either
$\vec(\widehat{\bU}) = (\bG^H\bG)^{-1} \bG^H \by$ such that $ \| \vec(  \widehat{\bU} ) \|_2^2  \leq d $
or
$\vec(\widehat{\bU}) \= (\bG^H\bG \+ \mu \bI_{d N_\mathrm{r}})^{-1} \bG^H \by$ such that  $\ g(\mu) \! \triangleq \! \| \vec(  \widehat{\bU} ) \|_2^2  \= d$,
where $\mu > 0$ is the unique solution to the fixed point equation $g(\mu) \= d$ and $g(\mu)$ is monotonically decreasing in $\mu>0$.
\end{lemma}
\begin{proof}
Suppose the affine map $\cA_{\cS}$ in \eqref{affine map}.
By the Kronecker product equality $\vec( \bW_k^H \bU \bSig \bV^H \bff_k) \= ( (\bSig \bV^H \bff_k)^T \otimes  \bW_k^H) \vec(\bU)$, $k=1, \ldots, K$ and collecting them as a long column vector yields an equivalent problem to \eqref{SVD method U omit} as
\beq
\widehat{\bU} \! = \! \argmin \limits_ {\bU} \lA \by \! - \! \bG\vec(\bU) \rA_2^2 ~\text{subject to~} \lA \vec( {\bU} ) \rA_2^2  \! \le \! d, \nn
\label{SVD method U new}
\eeq
where $\bG$ is given in \eqref{matrix B}. This problem is a standard differentiable convex optimization problem, which can be solved by applying Karush-Kuhn-Tucker (KKT) conditions \cite{boyd}.
Since it is a standard procedure, we omit the details here.
\end{proof}

We now turn our attention to the subspace power allocation problem (S3). To obtain a closed-form solution to (S3), we have the following lemma.
\begin{lemma}\label{t2}
	Consider the following quadratic program,
	\beq
	\begin{aligned}
		\! \widehat{\bSig}   \! =  \argmin \limits_ {\bSig}  \! \big \| \by  \!-\! \cA_{\cS}(\bU \bSig \bV^H) \big\|_2^2,\\
		\text{subject to~}  \lA  \bSig  \rA _F^2 \le \beta, \label{SD method lamd omit}
	\end{aligned}
	\eeq
	where $\by\in\C^{M\times 1}$, $\bU\in \C^{N_\mathrm{r} \times d}$, $\bV\in \C^{N_\mathrm{t} \times d}$, and $\bSig \in \C^{d \times d}$ is  diagonal.
	Define the matrix  $\bP \in \C^{M \times d}$ with the $i$th column
	\beq
	[\bP]_{:,i} = \cA_{\cS}([\bU]_{:,i}[\bV]_{:,i}^H). \label{matrix P}
	\eeq
	The  solution $\widehat{\bSig}$ is given by either
	$\diag(\widehat{\bSig}) = (\bP^H\bP)^{-1} \bP^H \by$ such that $ \| \diag(  \widehat{\bSig} ) \|_2^2  \leq \beta $
	or
	$\diag(\widehat{\bSig}) \= (\bP^H\bP \+ \rho \bI_{d })^{-1} \bP^H \by$ such that  $\ g(\rho) \! \triangleq \! \| \diag(  \widehat{\bSig} ) \|_2^2  \= \beta$,
	where $\rho > 0$ is the unique solution of the equation $g(\rho) \= \beta$.
\end{lemma}
\begin{proof}
        Rewrite $\bU \bSig \bV^H \=\sum_{i=1}^d [\bU]_{:,i} \! [\bSig]_{i,i} [\bV]_{:,i}^H$. Then,
	\beq
	\cA_{\cS}(\bU \bSig \bV^H) 
	\d4&=&\d4 \sum_{i=1}^{d}[\bSig]_{i,i} \cA_{\cS}([\bU]_{:,i}[\bV]_{:,i}^H). \label{lemma 5.1}
	\eeq
	Defining $[\bP]_{:,i} = \cA_{\cS}([\bU]_{:,i}[\bV]_{:,i}^H)$ in \eqref{lemma 5.1}, the problem in \eqref{SD method lamd omit} is given by
	\beq
	\begin{aligned}
		\! \diag(\widehat{\bSig}  ) \! =  \argmin_ {\bSig}  \! \big \| \by  \!-\! \bP \diag(\bSig) \big\|_2^2,\\
		~~\text{subject to~}  \lA  \bSig  \rA _F^2 \le \beta.  \label{SD method lamd omit1}
	\end{aligned}
	\eeq
	The optimal solution to \eqref{SD method lamd omit1} is found  by applying the KKT conditions \cite{boyd}.
\end{proof}

Notice that since the subsequent minimizations in (S1), (S2), and (S3) are optimally solved, a sequence of  $\{ \| \by - \cA_{\cS}(\bU_{\ell} \bSig_{\ell} \bV_{\ell}^H) \|_2^2 \}_{\ell \geq 1}$  is non-increasing and converges to a stationary point over the iterations in Algorithm \ref{alg1}.



\section{Simulation Results} \label{section6}
\subsection{Simulation Setup}
In this section, we evaluate the performance of low-rank subspace decomposition (SD) in Algorithm \ref{alg1}.
We adopt the physical representation of sparse millimeter wave MIMO channels \cite{Heath16, alk}, where $L$ scatters are assumed to constitute the propagation paths between the transmitter and receiver,
\beq
\bH = \sqrt{\frac{N_{{r}} N_{t}}{L}} \sum\limits_{l = 1}^L h_l \ba_r ({\xi_{r,l}})  \ba^H_t ({\xi_{t,l}}).  \label{channel_model}
\eeq
In \eqref{channel_model}, $h_{l}\in\C$ is the complex gain of the $l$th path, i.i.d. according to $h_{l}\sim \cC\cN(0, 1)$, and $\xi_{t,l}$ and $\xi_{r,l}$ are the angles of departures (AoDs) and arrivals (AoAs) at the transmitter and receiver, respectively.
Both $\xi_{t,l}$ and $\xi_{r,l}$ are i.i.d. and uniform over $[-\pi/2, \pi/2]$.
The $\ba_t(\xi_{t,l}) \!\in\! \C^{N_t\times 1}$ and $\ba_r (\xi_{r,l}) \!\in\! \C^{N_r\times 1}$ are the array response vectors at the transmitter and receiver, respectively.
We  assume, for simplicity, $\ba_t(\xi_{t,l})$ and {$\ba_r(\xi_{r,l})$} are uniform linear arrays (ULAs) with the antenna spacing set to half of the wavelength.
The channel model in \eqref{channel_model} satisfies $E[  \| \bH \|_F^2  ] \= N_rN_t$.
Throughout the simulation, we assume $N_r\=36$ and $N_t\=144$,  and the parameter $\beta$ in \eqref{eq_objective with P} for the SD method is set to $\beta=E[  \| \bH \|_F^2  ] \= N_rN_t$.

\begin{figure}[t]
	\centering
	\includegraphics[width=4.5in, height=3.4in]{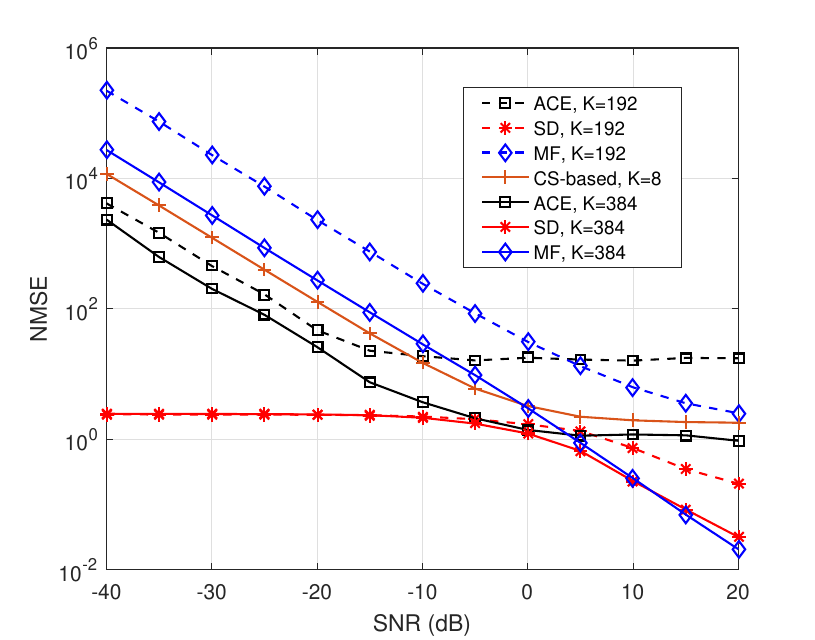}
	\caption{NMSE  vs. SNR (dB) ($N_t=144, N_{r}=36, L=4, d=4, N =4$).} \label{basic_MSE}
\end{figure}

We consider four other estimation techniques as our benchmarks, two open-loop channel estimation technique, i.e., the MF in \cite{jain2013low} and CS-based method in \cite{csChannel}, and two adaptive channel estimation techniques, i.e.,  the two-way Arnoldi approach in \cite{hadi2015} and closed-loop adaptive channel estimation (ACE) in \cite{alk}.
\redd{The work in \cite{csChannel} provides an open-loop compressed sensing (CS) approach.
Unlike the proposed SD, the CS-based method in \cite{csChannel} leverages the sparse representation of millimeter wave channels in angular domain. Hence, in \cite{csChannel}, the orthogonal matching pursuit (OMP) is employed to recover the angle supports of the channel matrix.
To enhance the sensing capability, it designs the measurement matrix to meet the incoherency property \cite{baraniuk2007compressive}.}
For the comparison, we admit two performance metrics, i.e., normalized MSE (NMSE) and  spectrum efficiency.
Specifically, the NMSE is defined as
\beq
\text{NMSE}=E  \left[ \lA \bH_d - \widehat{\bH} \rA_F^2 / \lA \bH_d \rA_F^2  \rS. \label{NMSE}
\eeq
The spectrum efficiency is defined as
\beq
R = \log_2 \left|\bI_d + \frac{1}{\sig^2} \bR_n^{-1}\tilde{\bH} {\tilde{\bH}}^H \right| , \label{spectrum efficiency}
\eeq
where $\tilde{\bH} \= \bW_D^H \bW_A^H \bH \bF_A\bF_D$ and $\bR_n \=  \bW_D^H\bW_A^H \bW_A \bW_D$ with $\bF_{A}\!\in \! \C^{N_t\times N}$, $\bF_{D}\!\in\! \C^{N\times d}$, $\bW_{A}\!\in\! \C^{N_r \times N}$, and $\bW_{D}\!\in\! \C^{N \times d}$, respectively, being the analog precoder, digital precoder, analog combiner, and digital combiner.
Here, we consider $d$-stream multiplexing.
For a fair comparison, we admit the near-optimal precoding technique proposed in  \cite{Ayach14}.
This is to say, given the channel estimate $\widehat{\bH}$ attained by the proposed SD, MF \cite{jain2013low}, \redd{CS-based method \cite{csChannel}}, ACE \cite{alk},  and Arnoldi \cite{hadi2015} approaches, the precoders and combiners are designed based on the technique in \cite{Ayach14} and then $R$ in \eqref{spectrum efficiency} for each approach is evaluated.
For the NMSE evaluation, we involve the proposed SD, MF, \redd{CS-based method}, and ACE\footnote{The Arnoldi approach in \cite{hadi2015}, is incapable of computing the NMSE in \eqref{NMSE}, since it only extracts the estimate of $\bH^H\bH$ or $\bH \bH^H$.},  while we consider all the benchmarks for the spectrum efficiency.

\begin{figure} [t]
	\centering
	\includegraphics[width=4.5in, height=3.4in]{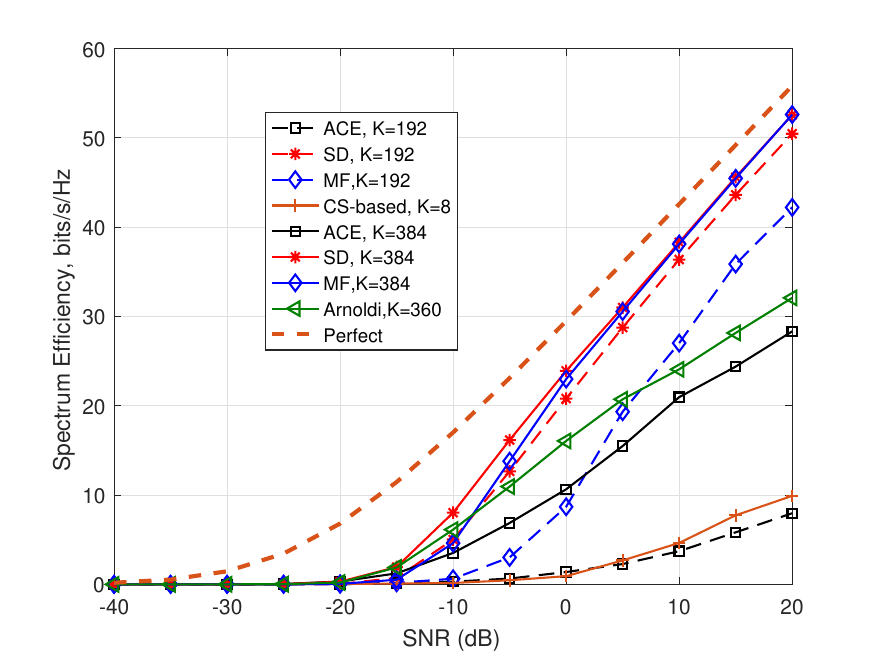}
	\caption{Spectrum efficiency   vs. SNR (dB) ($N_{t}=144, N_{r}=36, L=4, d=4, N =4$).} \label{basic_Rate}
\end{figure}

\subsection{Numerical Results}
In Fig. \ref{basic_MSE}, the NMSE of the proposed SD is compared with MF, \redd{CS-based method}, and ACE.
We assume the number of paths $L\=4$, the rank of the desired channel estimate $d\=4$, and the number of RF chains $N\= 4$.
{According to Theorem \ref{RIP measurements pre}, channel use overhead $K$ satisfies $K \!\geq\! {2\frac{L}{N}(N_r+N_t+1)}$ to meet the RIP.
We set the channel uses for SD, MF and ACE at $K\=192$ and $K\=384$, which are approximately half and equal to the bound defined in \eqref{minimal obser}, respectively.}
\redd{Because of the assumed angle grid channel model \cite{csChannel}, the number of channel uses
for the CS-based method is reduced to $K = 2N$, where $N$ is the number of RF chains.}
Seen from Fig. \ref{basic_MSE}, the proposed SD results in the most accurate channel estimation performance.
{Moreover, the NMSE of the proposed SD} shows sharp improvement at low SNR and it further decreases as the SNR grows, which is consistent with the analysis in \eqref{mmseb}.
\redd{Specifically, when the SNR is high, i.e., $0$ dB $\sim$ $20$ dB, and the number of channel uses is sufficiently large, i.e., $K=384$, MF and SD can both achieve accurate estimation.
The latter scenario represents an ideal case at high SNR with a sufficiently large number of observations.
On the other hand, when $K=192$, the MF experiences large performance degradation, while the proposed SD is still compatible with the case of $K=384$, demonstrating its robustness.
For ACE and CS-based method, even if the SNR and/or the number of channel uses are high, they can not provide an accurate channel estimate.
This happens because the channel estimation of these two methods is based on the angle grids \cite{alk}\cite{csChannel}.
When the true AoDs and AoAs do not fall exactly on the grid points, these two methods always suffer from residual NMSE even at very high SNR, thus the NMSE flooring in Fig. \ref{basic_MSE}.}
\redd{
When the SNR is less than $0$ dB in Fig. \ref{basic_MSE}, the performance of SD does not improve by increasing the number of channel uses $K$. This phenomenon is consistent with the mean square error bound analysis in \eqref{mmseb}. When the SNR is low, the NMSE in \eqref{mmseb} is noise-limited and dominated by $2\beta$, which is irrespective of $K$.}
\redd{As a result}, the proposed SD significantly outperforms MF, {especially at low SNR},  due to its exploitation of sparse subspace and power regularization.

In Fig. \ref{basic_Rate}, the spectrum efficiency of the proposed SD, MF, ACE, \redd{CS-based method}, and Arnoldi approaches are illustrated under the same simulation settings as in Fig. \ref{basic_MSE}.
The Perfect curve in Fig. \ref{basic_Rate} is based on perfect CSI and fully digital precoding to evaluate the spectrum efficiency.
Given the system dimension, the minimum channel use overhead required for Arnoldi \cite{hadi2015} is given by $K\=360$.
When the $K$ for ACE increases to $K\=384$, its performance becomes compatible with Arnoldi.
Seen from Fig. \ref{basic_Rate}, the proposed SD achieves the best spectrum efficiency and both MF and SD closely approach perfect CSI performance as SNR grows.
The proposed technique outperforms Arnoldi by only consuming about $50\%$ of the channel use overhead of Arnoldi.
{Seen from Fig. \ref{basic_MSE} and \ref{basic_Rate}, the performance gap of the proposed SD between $K\=192$ and $K\=384$ is slight and half of the bound in \eqref{minimal obser}, i.e., $K\=192$ can still offer near optimal performance.}
\redd{
Note that in Fig. \ref{basic_Rate}, though CS-based method requires only $K = 8$ channel
uses, the spectrum efficiency achieved by it is very low. This is because the CS-based method, assuming the discrete AoDs and AoAs as their channel priors, can not be extended to the practical channel model assumed in \eqref{channel_model}. Thus, in what follows, we shall only evaluate ACE, Arnoldi, MF, and the proposed SD, and exclude the CS-based method in the rest of the simulations.}
\redd{In Fig. \ref{basic_MSE}, the observation was that when the SNR is in between $-10$ dB and $5$ dB for $K=384$, the NMSE performance gap between ACE and SD is minor. However, this is not the case when comparing the spectrum efficiency in Fig. \ref{basic_Rate}.
In Fig. \ref{basic_Rate}, the proposed SD achieves considerably improved spectrum efficiency than the ACE. This is the direct consequence of the improved subspace estimation of SD.}
\begin{figure}[t]
	\centering
	\includegraphics[width=4.5in, height=3.4in]{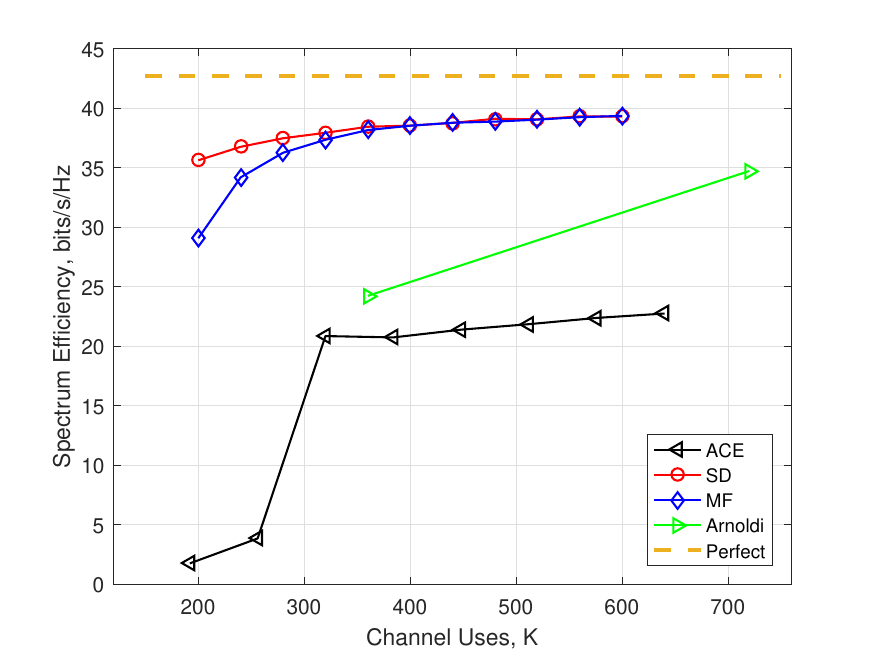}
	\caption{Spectrum efficiency  vs. Channel uses ($N_{t}=144, N_{r}=36,  d=4, L=4, N =4, \text{SNR} = 10~\text{dB}$).} \label{Channel_uses_Rate}
\end{figure}

In Fig. \ref{Channel_uses_Rate}, we evaluate the impact of the number of channel uses $K$ on the spectrum efficiency under the same simulation setting as before.
The curves are displayed at SNR $10$ dB across various $K$ values.
There is a clear performance improvement for ACE when $K\=384$ and above.
However, ACE requires considerably larger $K$ to further enhance its performance.
In the range of $K$ in Fig. \ref{Channel_uses_Rate}, the Arnoldi approach \cite{hadi2015} can only set $K\=360$ and $K\=720$.
It should be noted that the Arnoldi approach becomes ideal if the system can be equipped with a larger number of RF chains (e.g., $N \= N_t/8$ in \cite{hadi2015}).
The MF and SD benefit from increasing $K$, which is consistent with our analysis in Theorem \ref{the4} and \eqref{mmseb}.
The proposed SD achieves the highest spectrum efficiency.
It is worth noting that when the number of channel uses is over $440$, the performance of MF becomes similar to SD.

In Fig. \ref{paths_MSE}, we evaluate the NMSE performance for different values of channel path $L$.
The simulation parameters are the same as before except that we now consider $L\!\in\! \{4,5,6\}$ for $K\=192$.
As can be seen from Fig. \ref{paths_MSE}, as $L$ grows from 4 to 6, the NMSE performance of SD is deteriorated.
This is due to the increased residual error $\lA \bH_{L \setminus d}\rA_F^2$ in \eqref{error bound}  as $L$ increases according to Theorem \ref{the4}.
The same trend can be observed for MF.
Nevertheless, the performance gap for different $L$ values is negligible at low SNR for SD.
It is observed that the number of paths does not have a considerable effect on the performance of ACE.

\begin{figure}[t]
	\centering
	\includegraphics[width=4.5in, height=3.4in]{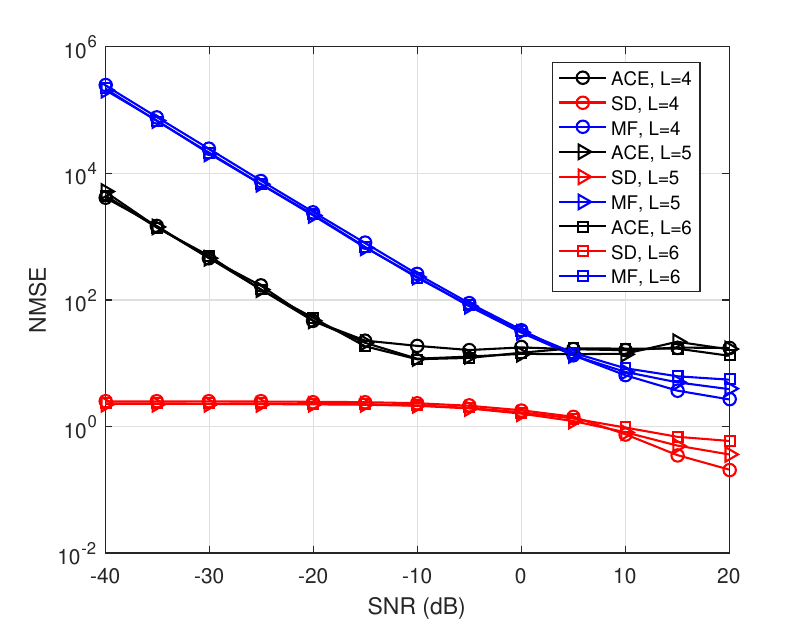}
	\caption{ NMSE  vs. SNR (dB) ($N_{t}=144, N_{r}=36,d=4, K =192, N=4, L=4,5,6$).} \label{paths_MSE}
\end{figure}

In this set of simulations, the effects of the number of RF chains $N$ to the NMSE and spectrum efficiency are evaluated.
In Fig. \ref{Nrf_MSE}, the NMSE curves for $N\! \in \!\{ 4,6 \}$ are displayed.
As illustrated in Fig. \ref{Nrf_MSE}, for fixed $K\=192$, the more RF chains the better NMSE performance is expected.
It can also be seen that the proposed SD is less sensitive to the change of $N$ at low SNR, while achieving the lowest NMSE.
The spectrum efficiency under the different numbers of RF chains is demonstrated in Fig. \ref{Nrf_Rate}.
In Fig. \ref{Nrf_Rate}, we set $K$ for MF and SD at $K\=192$ and ACE at $K\=192$ and $K\=384$.
{For Arnoldi, when $N\=4$, the minimum required channel use is $K\=360$, while when $N\=6$, this is be decreased to $K\=240$.}
The number of channel uses for MF and SD are the least among all the benchmarks.
Nevertheless, the proposed SD still shows the best performance.

\begin{figure}[t]
	\centering
	\includegraphics[width=4.5in, height=3.4in]{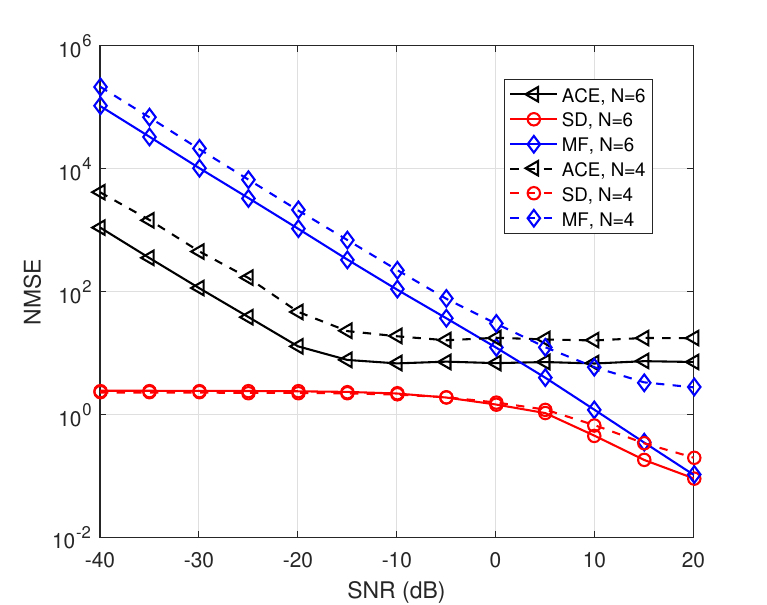}
	\caption{ NMSE  vs. SNR (dB) ($N_{t}=144, N_{r}=36,  L=4,d=4, K =192, N=4,6$).} \label{Nrf_MSE}
\end{figure}

\redd{
\subsection{Complexity Analysis}}
\redd{In this subsection, we compare the computational complexities of the proposed SD, ACE \cite{alk}, MF \cite{Haldar_factor}, and Arnoldi \cite{hadi2015} in terms of the number of complex scalar multiplications.
Provided the equal channel uses $K$, the complexity introduced by the precoding and combining is the same for all the approaches.
Therefore, we only need to compare the complexities involved in the alternating minimizations for SD and MF, the two-way echoing vector adaptation for Arnoldi \cite{hadi2015}, and the training precoder and combiner adaptation for ACE \cite{alk}.
Specifically, for the proposed SD, the computational complexity is determined by the maximum number of iterations $\ell_{max}$ for (S1), (S2), and (S3) and the associated computations in Algorithm 1.
Since these three subproblems are all least squares over a sphere, the SVD constitutes  the dominant computational overhead.
Hence, for a fixed $\ell_{max}$, the complexity of SD is given by ${\cal{O}}\! \left( K N d^2 (N_t^2+N_r^2)\right)$.
The same overhead ${\cal{O}}\! \left(K N d^2 (N_t^2+N_r^2) \right)$ results for MF.
The number of AoD and AoA search levels of ACE in \cite{alk} is given by ${\cal{O}} (KN/d^3)$.
At each search level of ACE, the complexity is dominated
by the computations of training precoders and combiners, which is given by ${\cal{O}}\left(  (N_r +N_t)D\right)$ \cite{alk}, where $D$ ($ \gg \max(N_t, N_r)$) denotes the cardinality of an over-complete dictionary.
Therefore, the complexity of the ACE is given by ${\cal{O}} \left(KND(N_r+N_t)/d^2 \right)$.
The main computational complexity of Arnoldi is the design of channel echoing vectors, which is given by ${\cal{O}}\left( K^2N^2/(N_r+N_t)\right)$ \cite{hadi2015}.
In summary, the Arnoldi method has the lowest complexity ${\cal{O}}\left( K^2N^2/(N_r+N_t)\right)$. The proposed SD has the similar complexity with MF and ACE, which is about ${\cal{O}}\! \left(K N d^2 (N_t^2+N_r^2)\right)$.}

\section{Conclusions} \label{section9}
In this paper, we have proposed an open-loop low-rank subspace decomposition  technique for millimeter wave MIMO channel estimation.
The RIP that ensures reliability of the low-rank matrix reconstruction was established for the channel subspace sampling signals randomly generated by the hybrid analog and digital arrays.
When the number of channel uses is around $\frac{L}{N} (N_t \+ N_r)$, the considered random subspace sampling signals satisfy the RIP with high probability.
The established RIP was applied to analyze the channel estimation error performance.
This analysis showed the resilience of the proposed technique at low SNR.
Moreover, we devised an alternating optimization algorithm that effectively finds a suboptimal but stationary solution to the problem.
The simulation studies corroborate our analysis and showed that the proposed low-rank subspace decomposition  technique can achieve near optimal performance at high SNR, while offering robust performance at low SNR with channel use overhead less than $\frac{L}{N}(N_r+N_t)$.
It outperforms other adaptive closed-loop and two-way channel estimation techniques with considerably reduced channel use overhead.
This demonstrates that a non-adaptive open-loop channel estimation scheme, if it is carefully designed, can still be used to provide accurate CSI for  large-scale millimeter wave MIMO channels.
The proposed techniques does not assume any explicit channel model and statistics, making it extendable to a wide range of sparse scenarios.

\begin{figure}[t]
	\centering
	\includegraphics[width=4.5in, height=3.4in]{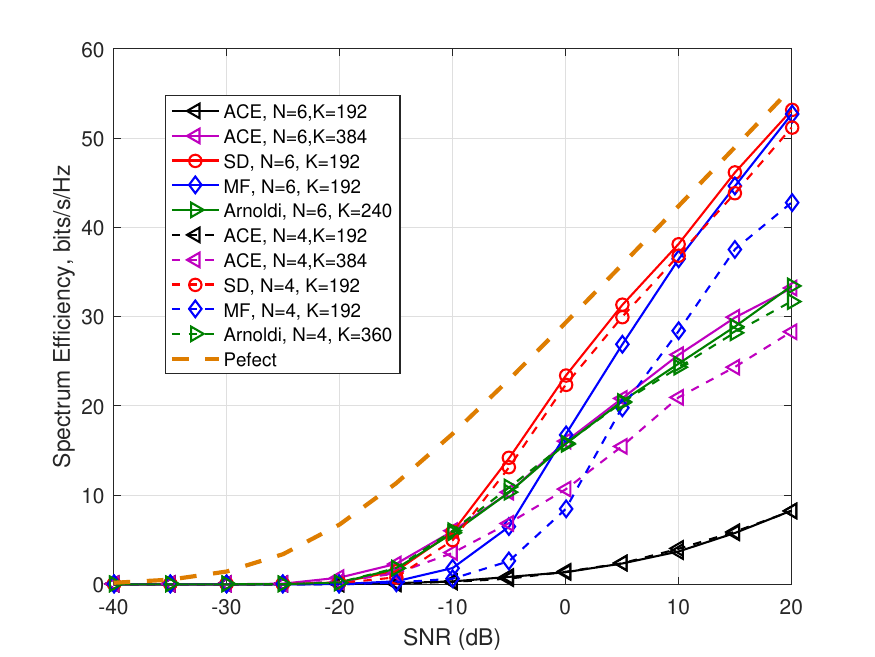}
	\caption{ Spectrum efficiency   vs. SNR (dB) ($N_{t}=144, N_{r}=36,  L=4,d=4, N=4,6$).} \label{Nrf_Rate}
\end{figure}

\appendices
\section{Proof of Lemma \ref{expectation lemma}} \label{appendix1}

Any phase rotation applied to $[\bA_{\cS}]_{m,n}$ does not change its distribution, i.e., for any $\varphi \in [0,2\pi)$,
\beq
[\bA_{\cS}]_{m,n} \stackrel{d}{=} [\bA_{\cS}]_{m,n} e^{j\varphi}. \nn
\eeq
Hence the following holds
\beq
E\left[ \left|[\bA_{\cS}]_{m,:}{\bx}\right|^{2u}\right]  = E\left[    \left|[\bA_{\cS}]_{m,:}\bar{\bx}\right|^{2u} \right]. \nn
\eeq
Both sides of \eqref{theorem2 in} can be expanded as
\beq
\begin{aligned}
	& \! E \! \left[  |[\bA_{\cS}]_{m,:}\bar{\bx}|^{2u}\right] \=\! \sum_{n_1 = 1}^{N_r N_t}\sum_{n_2 = 1}^{N_r N_t} \!\!\! \cdots \!\!\!\! \sum_{n_{2u-1} = 1}^{N_r N_t} \sum_{n_{2u} = 1}^{N_r N_t}  E\big[  [\bA_{\cS}]_{m,n_1} [\bar{\bx}]_{n_1}   \\
	& ~\times [\bA_{\cS}]_{m,n_2}^H [\bar{\bx}]_{n_2} \cdots [\bA_{\cS}]_{m,n_{2u-1}} [\bar{\bx}]_{n_{2u-1}} [\bA_{\cS}]_{m,n_{2u}}^H [\bar{\bx}]_{n_{2u}} \big]
	\nn
\end{aligned}
\eeq
and
\beq
\begin{aligned}
	&\!\! E\!\left[  |[\bB]_{m,:}\bar{\bx}|^{2u} \right] \!\=\!\!
	\sum_{n_1 = 1}^{N_r N_t} \!\! \cdots \!\!\!\! \sum_{n_{2u} = 1}^{N_r N_t}\!\!\! E\!\left[  [\bB]_{m,n_1} [\bar{\bx}]_{n_1} \!\!\cdots [\bB]_{m,n_{2u}} [\bar{\bx}]_{n_{2u}} \right], \nonumber
\end{aligned}
\eeq
respectively.
Since $[\bar{\bx}]_n \ge 0$, $n= 1,\ldots,N_rN_t$, in order to verify \eqref{theorem2 in}, it is sufficient to show
\beq
E\left[  [\bA_{\cS}]_{m,n_1} \ldots [\bA_{\cS}]_{m,n_{2u}}^H \right] \le
E\left[   [\bB]_{m,n_1} \ldots [\bB]_{m,n_{2u}}  \right]  \label{expectation inequality}
\eeq
for all $2u$-tuples $\{ (n_1, \ldots, n_{2u}) \}$.

We rewrite the inequality in  \eqref{expectation inequality} as
\beq
E \Bigg[ \! \prod_{n=1}^{N_tN_r} \! ([\bA_{\cS}]_{m,n})^{p_n} ([\bA_{\cS}]^H_{m,n})^{q_n} \! \Bigg]
\! \le \! E  \Bigg[   \prod \limits _{n=1}^{N_tN_r} ([\bB]_{m,n})^{p_n+q_n} \!\Bigg ], \nn
\eeq
where $p_n$ and $q_n$ are non-negative integers such that  $\sum_{n=1}^{N_rN_t}p_n = \sum_{n=1}^{N_rN_t}q_n=u$.
Since the entries of $\bA_{\cS}$ and $\bB$ are i.i.d., the latter is equivalent to
\beq
\prod \limits _{n=1}^{N_tN_r} \!\!  E \big[  \!\left([\bA_{\cS}]_{m,n}\right)^{p_n}  \! \left([\bA_{\cS}]^H_{m,n} \right)^{q_n} \! \big]\!
\d4&\le&\d4 \! \prod \limits _{n=1}^{N_tN_r} \!\!  E\big[ ( [\bB]_{m,n})^{p_n+q_n}\big]. \nn \\
& & \label{expectation inequality compact2}
\eeq
Note that both sides of \eqref{expectation inequality compact2} can be either $\frac{1}{(N_r N_t)^{u}}$ or $0$ depending on the $p_n$ and $q_n$ values.
When $\prod  _{n=1}^{N_tN_r} E [  \left([\bA_{\cS}]_{m,n}\right)^{p_n}  \left([\bA_{\cS}]^H_{m,n} \right)^{q_n} ] \=\frac{1}{(N_rN_t)^{u}}$, we must have $p_n\=q_n$, $\forall n$. This is to say that the left hand side (l.h.s.) of \eqref{expectation inequality compact2} is $E\big[  \prod _{n=1}^{N_r N_t} \left|[\bA_{\cS}]_{m,n}\right|^{2p_n} \big] \= \frac{1}{(N_r N_t)^{u}}$,
if and only if $p_n\=q_n$, $\forall n$.
When $p_n\=q_n$, it is easy to see that the right hand side (r.h.s.) of \eqref{expectation inequality compact2} is also $\E [ \prod_{n=1}^{N_r N_t} |[\bB]_{m,n}|^{2p_n} ] \= \frac{1}{(N_r N_t)^{u}}$.
On the other hand, when $p_n \neq q_n$, the l.h.s. of \eqref{expectation inequality compact2} is
$\prod  _{n=1}^{N_tN_r} E [ \left([\bA_{\cS}]_{m,n}\right)^{p_n}  \left([\bA_{\cS}]^H_{m,n} \right)^{q_n} ] =0$ and the inequality in \eqref{expectation inequality compact2} always holds regardless of the value on the r.h.s., i.e., when the l.h.s. of \eqref{expectation inequality compact2} is zero, the r.h.s. can be either $\frac{1}{(N_r N_t)^{u}}$ or $0$.
This concludes the proof.

\section{Proof of Theorem \ref{unit sensing matrix}} \label{appendix2}

Since a part of the proof is based on a similar technique in \cite{database},  we first introduce the following lemma, which will be used in the proof.

\begin{lemma}[\cite{database}]
	Suppose $\bb \!\in \! \R^{N_r N_t \times 1}$ is a random vector where $[\bb]_i \! \in\!  \{ {-1}/{\sqrt{N_r N_t}}, {1}/{\sqrt{N_r N_t}}\}$ with equal probability. Then for any $h \!\in\! [0,N_r N_t/2]$, and any unit norm real vector $\ba \!\in\! \R^{N_r N_t \times 1}$, $\| \ba \|_2^2\=1$,
	\beq
	E \lS  \exp \left( h |\bb^T \ba |^2  \right) \rS \le \frac{1}{\sqrt{1-2h/(N_r N_t)}} \label{database eq1}
	\eeq
	and
	\beq
	E\left[  |\bb^T \ba |^4 \right] \le \frac{3}{(N_r N_t)^2}. \label{database eq2}
	\eeq
\end{lemma}

We first focus on the upper tail probability in \eqref{upper tail in}, i.e., $\Pr ( \| \sqrt{N_t N_r/M}\bA_{\cS} \bx \|_2^2   >  1+\al )$.
For any $h>0$, applying Markov's inequality to the upper tail probability yields
\beq
\d4&\Pr&\!\!\d4 \big( \| \sqrt{N_t N_r/M}\bA_{\cS} \bx \|_2^2   >  1+\al \big) \! \leq \! E \Big[ e^{h \| \bA_{\cS} \bx \|_2^2} \Big]  e^{- h\frac{M(1+\al)}{N_r N_t}} \nn \\
\d4&&\d4 \hspace{2.5cm} = \lp E \lS e^{h | [\bA_{\cS}]_{1,:} \bx |^2}  \rS \rp^M e^{- h\frac{M(1+\al)}{N_r N_t}}, \label{theorem 2 eq1}
\eeq
where the last step follows from the fact that the entries of $\bA_{\cS}$ are i.i.d.
The term $E [e^{h |  [\bA_\cS]_{1,:} \bx |^2}]  $ in \eqref{theorem 2 eq1} can be further upper bounded by
\beq
E \lS e^{h \la  [\bA_{\cS}]_{1,:} \bx \ra^2} \rS
\d4 &\leq &\d4 E \lS  e^{h |[\bB]_{1,:}\bar{\bx}|^2 } \rS \nn \\
\d4 &{\leq}& \d4 \frac{1}{\sqrt{1-2h/(N_r N_t)}}, \label{single expectation}
\eeq
where the first inequality is due to the Taylor series expansion of   $E\big[ e^{h |  [\bA_\cS]_{1,:} \bx |^2} \big]= \sum_{k = 0}^{\infty} \frac{h^k}{k!}  E [ | [\bA_{\cS}]_{1,:}{\bx}|^{2k} ]$ and the inequality $E [ | [\bA_{\cS}]_{1,:}{\bx}|^{2k} ]\leq E [ | [\bB]_{1,:}\bar{\bx}|^{2k} ]$ in \eqref{theorem2 in} of Lemma \ref{expectation lemma}, with  $\bB$ following the definition in Lemma \ref{expectation lemma}.
The last step in \eqref{single expectation} is due to \eqref{database eq1}.
Inserting \eqref{single expectation} into \eqref{theorem 2 eq1} yields
\begin{eqnarray}
\d4\d4\d4\d4&\Pr&\!\!\d4 \big( \| \sqrt{N_t N_r/M}\bA_{\cS} \bx \|_2^2   > 1+\al \big) \nn\\
\d4\d4\d4\d4&& \hspace{2cm} \leq \lp 1\-2h/(N_r N_t) \rp^{-\frac{M}{2}} e^{- h\frac{M(1+\al)}{N_r N_t}}.\label{temm eq}
\end{eqnarray}
Since the inequality holds for any $h\!>\!0$,
insearting $h \= \frac{N_r N_t \al}{2(1+\al)}$ in \eqref{temm eq} gives
\beq
\Pr \! \big( \| \sqrt{N_t N_r/M}\bA_{\cS} \bx \|_2^2   > 1+\al \big) \d4&\leq&\d4 \left( {1+\al}   \right)^{\frac{M}{2}} e^{-\al \frac{M}{2}} \nn \\
\d4& \leq & \d4 e^{-\frac{M}{2}\lp \frac{\al^2}{2} -\frac{\al^3}{3} \rp} , \nn
\eeq
where the last inequality  follows from $( 1+\al )^{\frac{M}{2}} e^{-\al \frac{M}{2}} \= e^{\frac{M}{2} (\ln(1+\al) -\al)}$ and $\ln(1+\al) \le \al - \al^2/2+\al^3/3, \forall \al \in (0,1)$.
This proves the upper tail bound in \eqref{upper tail in}.

Using the same approaches in \eqref{theorem 2 eq1}, the lower tail bound of \eqref{upper tail in} is upper bounded by
\beq
\d4\d4&\Pr&\!\!\d4 \big( \| \sqrt{N_t N_r/M}\bA_{\cS} \bx \|_2^2   <  1-\al  \big) \nn\\
\d4\d4&&\d4 \hspace{2cm} \leq \lp  E \lS  e^{-h | [\bA_{\cS}]_{1,:} \bx |^2} \rS \rp^M  e^{h\frac{M(1-\al)}{N_r N_t}} . \label{theorem 2 eq2}
\eeq
Applying the Taylor series expansion of $E \big[  e^{-h | [\bA_{\cS}]_{1,:} \bx |^2} \big]$ and taking the first three dominant terms yields
\beq
E \lS  e^{-h | [\bA_{\cS}]_{1,:} \bx |^2} \rS \! \leq \! E \bigg[ \! 1 \- h | [\bA_{\cS}]_{1,:} \bx |^2 \+ \frac{h^2 | [\bA_{\cS}]_{1,:} \bx |^4}{2} \! \bigg] \nn.
\eeq
Now, using the fact that $E[ | [\bA_{\cS}]_{1,:} \bx |^2] \= \frac{1}{N_rN_t}$ and applying a sequence of  bounds $E  [ | [\bA_{\cS}]_{1,:} \bx |^4 ] \leq E  [ | [\bB]_{1,:} \bar{\bx} |^4 ] \leq \frac{3}{(N_r N_t)^2}$, where the first bound is due to Lemma \ref{expectation lemma} and the second one is due to \eqref{database eq2}, yields
\beq
E \lS  e^{-h | [\bA_{\cS}]_{1,:} \bx |^2} \rS \leq  1 \- \frac{h}{N_r N_t} \+ \frac{3h^2}{2(N_r N_t)^2}. \label{theorem 2 eq3}
\eeq
Then, plugging in \eqref{theorem 2 eq3} into \eqref{theorem 2 eq2} with the substitution  $h = \frac{N_r N_t \al}{2(1+\al)}$ results in
\beq
\d4\d4\d4\d4&\Pr&\!\!\d4 \big( \| \sqrt{N_t N_r/M}\bA_{\cS} \bx \|_2^2   \leq  1-\al \big) \nn\\
\d4\d4\d4\d4&&\d4 \hspace{1.5cm} \leq \bigg(1 \-\frac{\al(4\+\al)}{8(1\+\al)^2} \bigg)^M e^{\frac{\al(1-\al)M}{2(1+\al)}}  \nn \\
\d4\d4\d4\d4&&\d4 \hspace{1.5cm} \leq e^{-\frac{M}{2}(\frac{\al^2}{2}-\frac{\al^3}{3})}, \label{ln expansion2}
\eeq
where the last inequality follows from $\big( 1 \-\frac{\al(4\+\al)}{8(1\+\al)^2} \big)^M e^{\frac{\al(1-\al)M}{2(1+\al)}} = e^{ M \big(  \ln\big( 1 -\frac{\al(4\+\al)}{8(1\+\al)^2} \big)  +  \frac{\al(1-\al)}{2(1+\al)} \big) }$ and the inequality $\ln \big(1\-\frac{\al(4+\al)}{8(1+\al)^2}\big)\leq - \frac{\al(1-\al)}{2(1+\al)} \- \frac{\al^2}{4} \+ \frac{\al^3}{6}$.
This concludes the proof.

\section{Proof of Theorem \ref{the4}} \label{appendix3}

We start by realizing that the desired channel $\bH_d= \sum_{l=1}^d \lam_l  \bpsi_l \bphi^H_l$ is  feasible to \eqref{eq_objective with P}.
Thus, one should have
\beq
\lA  {\by} - \cA_{\cS}\big(\widehat{\bH}\big) \rA_2 \le \lA {\by} - \cA_{\cS}(\bH_d) \rA_2. \label{optimal for rank}
\eeq
Inserting ${\by} \= \cA_{\cS}(\bH_d) \+ \cA_{\cS}(\bH_{L \setminus d}) \+ \tilde{\bn}$ in \eqref{optimal for rank} yields
\beq
\big\| \cA_{\cS}\big({\bH_d} -\widehat{\bH}  \big)\! +\! \cA_{\cS}(\bH_{L \setminus d}) \!+ \! \tilde{\bn} \big\|_2 \! \le \! \lA \cA_{\cS}(\bH_{L \setminus d})\! + \! \tilde{\bn} \rA_2. \label{tr1}
\eeq
Applying the triangular inequality to the l.h.s. of \eqref{tr1} gives
\beq
\begin{aligned}
	& \big\| \cA_{\cS}\big( {\bH_d} - \widehat{\bH}\big)\! +\! \cA_{\cS}(\bH_{L \setminus d}) \!+ \! \tilde{\bn} \big\|_2 \\
	& \hspace{1.8cm} \geq \big\| \cA_{\cS}(\widehat{\bH} - {\bH_d})  \big\|_2 - \lA  \cA_{\cS}(\bH_{L \setminus d}) +  \tilde{\bn}  \rA_2.\label{tr2}
\end{aligned}
\eeq
Thus, combining \eqref{tr1} and \eqref{tr2} results in
\beq
\big\| \cA_{\cS}(\widehat{\bH} - {\bH_d})  \big\|_2  \le 2 \lA \cA_{\cS}(\bH_{L \setminus d}) +  \tilde{\bn} \rA_2. \label{theorem 4.1}
\eeq
Note that $\rank(\widehat{\bH} - {\bH_d})\leq 2d$. Then, based on the RIP of $\sqrt{\frac{N_r N_t}{M}}\cA_{\cS}(\cdot)$ by Theorem \ref{RIP measurements pre}, we have
\beq
\begin{aligned}
	(1-\delta_{2d}) \frac {M}{N_t N_r} \big\|  \bH_d - \widehat{\bH} \big\|_F^2  \le \big\| \cA_{\cS}(\widehat{\bH} - {\bH_d})  \big\|_2^2 \label{theorem 4.2},
\end{aligned}
\eeq
where $\delta_{2d}$ is $2d$-RIP constant.
Hence, combining \eqref{theorem 4.1} and \eqref{theorem 4.2} results in  the error bound
\beq
\big\| \bH_d - \widehat{\bH} \big\|_F^2 \le \frac{4 N_t N_r} {(1-\delta_{2d}) M}   \lA \cA_{\cS}(\bH_{L \setminus d}) +  \tilde{\bn} \rA_2^2. \label{Theorem4 conclusion 1}
\eeq
The inequality constraint in \eqref{eq_objective with P} is equivalent to  $\| \widehat{\bSig} \|_F^2 \le \beta$.
Moreover, since $\bH_d$ is feasible to \eqref{eq_objective with P}, one should have by the triangular inequality
\beq
\big\| \bH_d - \widehat{\bH} \big\|_F^2 \le 2\beta.\label{Theorem4 conclusion 2}
\eeq

Finally, combining \eqref{Theorem4 conclusion 1} and \eqref{Theorem4 conclusion 2} leads to
\beq
\big\| \bH_d - \widehat{\bH} \big\|_F^2  \le  \min \left \{\frac{4 N_t N_r \| \cA_{\cS}(\bH_{L \setminus d}) +  \tilde{\bn} \|_2^2} {(1-\delta_{2d}) M}   , 2\beta \right\}. \nn
\eeq
This bound holds, by Theorem \ref{RIP measurements pre}, with probability greater than or equal to $1-\exp(-qM)$ when $M\geq 4d(N_t+N_r+1)$ for $q = \frac{\al^2}{4}-\frac{\al^3}{6}-\frac{\ln(36 \sqrt{2}/\del_{2d})}{2}$ and $0<\alpha<1$.

\bibliographystyle{IEEEtran}
\bibliography{IEEEabrv,Conference_mmWave_CS}

\clearpage

\end{spacing}
\end{document}